%% file: draft.tex
\documentclass[11pt,letterpaper]{article}

\usepackage{graphicx}
\usepackage{amssymb}
\usepackage{amsthm}
\usepackage{amsmath}
\usepackage[margin=1in]{geometry}
\usepackage{placeins}
\usepackage[font={small,it}]{caption}
\usepackage{subfigure}
\usepackage{authblk}
\usepackage{algorithm,algorithmic}

\usepackage{titletoc}

\newcommand{\F}{\mathcal{I}}
\newcommand{\s}{\mathfrak{H}}

\newcommand{\colspace}{\mathrm{column}\mbox{-}\mathrm{space}}
\newcommand{\rowspace}{\mathrm{row}\mbox{-}\mathrm{space}}
\newcommand{\rk}{\mathrm{rank}}
\newcommand{\train}{\mathrm{train}}
\newcommand{\test}{\mathrm{test}}

\numberwithin{equation}{section}
\theoremstyle{plain}
\newtheorem{theorem}{Theorem}[section]
\newtheorem{corollary}[theorem]{Corollary}

\newtheorem{proposition}[theorem]{Proposition}

\newcommand{\R}{\mathbb{R}}
\newcommand{\A}{\mathcal{A}}
\newcommand{\Sp}{\mathbb{S}}
\DeclareMathOperator*{\argmin}{\arg\!\min}
\newcommand{\M}{\mathcal{M}}
\newcommand{\B}{\mathcal{B}}

\title{Interpreting Latent Variables in Factor Models via Convex Optimization}
\author{Armeen Taeb $^\dag$ and Venkat Chandrasekaran $^\dag, ^\ddag$
\thanks{Email: ataeb@caltech.edu, venkatc@caltech.edu} \vspace{.25in} \\ $^\dag$ Department of Electrical Engineering \\  $^\ddag$ Department of Computing and Mathematical Sciences \\ California Institute of Technology \\ Pasadena, Ca 91125}

\begin{document}
\maketitle

\begin{abstract}
Latent or unobserved phenomena pose a significant difficulty
in data analysis as they induce complicated and confounding
dependencies among a collection of observed variables.  Factor
analysis is a prominent multivariate statistical modeling approach
that addresses this challenge by identifying the effects of (a small
number of) latent variables on a set of observed variables.  However,
the latent variables in a factor model are purely mathematical objects
that are derived from the observed phenomena, and they do not have any
interpretation associated to them.  A natural approach for attributing
semantic information to the latent variables in a factor model is to
obtain measurements of some additional plausibly useful covariates
that may be related to the original set of observed variables, and to
associate these auxiliary covariates to the latent variables.  In this
paper, we describe a systematic approach for identifying such
associations.  Our method is based on solving computationally
tractable convex optimization problems, and it can be viewed as a
generalization of the minimum-trace factor analysis procedure for
fitting factor models via convex optimization.  We analyze the
theoretical consistency of our approach in a high-dimensional setting
as well as its utility in practice via experimental demonstrations
with real data. 
\end{abstract}

\input{Introduction}
\input{TheoreticalResults}

\input{Experimental}
\input{ProofStrategy}

\section{Discussion} In this paper we describe a new approach for interpreting
the latent variables in a factor model.  Our method proceeds by
obtaining observations of auxiliary covariates that may plausibly be
related to the observed phenomena, and then suitably associating these
auxiliary covariates to the latent variables.  The procedure involves
the solutions of computationally tractable convex optimization
problems, which are log-determinant semidefinite programs that can be
solved efficiently.  We give both theoretical as well as experimental
evidence in support of our methodology.  Our technique generalizes
transparently to other families beyond factor models such as
latent-variable graphical models \cite{Chand2012}, although we do not pursue
these extensions in the present article.

\newpage

\input{Appendix}

\end{document}

%% file: Introduction.tex
\section{Introduction}
A central goal in data analysis is to identify concisely described models that characterize the statistical dependencies among a collection of variables. Such concisely parametrized models avoid problems associated with overfitting, and they are often useful in providing meaningful interpretations of the relationships inherent in the underlying variables. Latent or unobserved phenomena complicate the task of determining concisely parametrized models as they induce confounding dependencies among the observed variables that are not easily or succinctly described. Consequently, significant efforts over many decades have been directed towards the problem of accounting for the effects of latent phenomena in statistical modeling. A common shortcoming of approaches to latent-variable modeling is that the latent variables are typically mathematical constructs that are derived from the originally observed data, and these variables do not directly have semantic information linked to them. Discovering interpretable meaning underlying latent variables would clearly impact a range of contemporary problem domains throughout science and technology. For example, in data-driven approaches to scientific discovery, the association of semantics to latent variables would lead to the identification of new phenomena that are relevant to a scientific process, or would guide data-gathering exercises by providing choices of variables for which to obtain new measurements.

In this paper, we focus for the sake of concreteness on the challenge of interpreting the latent variables in a factor model \cite{Spearman}. \emph{Factor analysis} is perhaps the most widely used latent-variable modeling technique in practice. The objective with this method is to fit observations of a collection of random variables $y \in \R^p$ to the following linear model:
\begin{equation}
y = \B \zeta + \epsilon,
\label{eqn:factormodel}
\end{equation}
where $\B \in \R^{p \times k}, k \ll p$. The random vectors $\zeta \in \R^k, \epsilon \in \R^p$ are independent of each other, and they are normally distributed as\footnote{The mean vector does not play a significant role in our development, and therefore we consider zero-mean random variables throughout this paper.} $\zeta \sim \mathcal{N}(0, \Sigma_\zeta), \epsilon \sim \mathcal{N}(0,\Sigma_\epsilon)$, with $\Sigma_\zeta \succ 0, \Sigma_\epsilon \succ 0$ and $\Sigma_\epsilon$ being diagonal. Here the random vector $\zeta$ represents a small number of unobserved, latent variables that impact all the observed variables $y$, and the matrix $\B$ specifies the effect that the latent variables have on the observed variables. However, the latent variables $\zeta$ themselves do not have any interpretable meaning, and they are essentially a mathematical abstraction employed to fit a concisely parameterized model to the conditional distribution of $y | \zeta$ (which represents the remaining uncertainty in $y$ after accounting for the effects of the latent variables $\zeta$) -- this conditional distribution is succinctly described as it is specified by a model consisting of independent variables (as the covariance of the Gaussian random vector $\epsilon$ is diagonal).

A natural approach to attributing semantic information to the latent variables $\zeta$ in a factor model is to obtain measurements of some additional plausibly useful covariates $x \in \R^q$ (the choice of these variables is domain-specific), and to link these to the variables $\zeta$.  However, defining and specifying such a link in a precise manner is challenging.  Indeed, a fundamental difficulty that arises in establishing this association is that the variables $\zeta$ in the factor model \eqref{eqn:factormodel} are not identifiable. In particular, for any non-singular matrix $\mathcal{W} \in \R^{k \times k}$, we have that $\B \zeta = (\B \mathcal{W}^{-1}) (\mathcal{W} \zeta)$. In this paper, we describe a systematic and computationally tractable methodology based on convex optimization that integrates factor analysis and the task of interpreting the latent variables. Our convex relaxation approach generalizes the \emph{minimum-trace factor analysis} technique, which has received much attention in the mathematical programming community over the years \cite{Ledermann,Shapiro1,S2,S3,Saunderson}.

\subsection{A Composite Factor Model}
\label{secion:composite}
We begin by making the observation that the column space of $\B$ -- which specifies the $k$-dimensional component of $y$ that is influenced by the latent variables $\zeta$ -- is invariant under transformations of the form $\B \rightarrow \B \mathcal{W}^{-1}$ for non-singular matrices $\mathcal{W} \in \R^{k \times k}$. Consequently, we approach the problem of associating the covariates $x$ to the latent variables $\zeta$ by linking the effects of $x$ on $y$ to the column space of $\B$. Conceptually, we seek a decomposition of the column space of $\B$ into transverse subspaces $\s_x, \s_u \subset \R^p, ~ \s_x \cap \s_u = \{0\}$ so that $\colspace(\B) \approx \s_x \oplus \s_u$ -- the subspace $\s_x$ specifies those components of $y$ that are influenced by the latent variables $\zeta$ and are also affected by the covariates $x$, and the subspace $\s_u$ represents any unobserved residual effects on $y$ due to $\zeta$ that are not captured by $x$. To identify such a decomposition of the column space of $\B$, our objective is to split the term $\B \zeta$ in the factor model~\eqref{eqn:factormodel} as
\begin{equation}
\B \zeta \approx \A x + \B_u \zeta_u,
\label{eqn:mid}
\end{equation}
where the column space of $\A \in \R^{p \times q}$ is the subspace $\s_x$ and the column space of $\B_u \in \R^{p \times \dim(\s_u)}$ is the subspace $\s_u$, i.e., $\dim(\colspace(\A)) \allowbreak + \dim(\colspace(\B_u)) \allowbreak = \dim(\colspace(\B))$ and $\colspace(\A) \allowbreak \cap \colspace(\B_u) \allowbreak = \{0\}$. Since the number of latent variables $\zeta$ in the factor model~\eqref{eqn:factormodel} is typically much smaller than $p$, the dimension of the column space of $\A$ is also much smaller than $p$; as a result, if the dimension $q$ of the additional covariates $x$ is large, the matrix $\A$ has small rank. Hence, the matrix $\A$ plays two important roles: its column space (in $\R^p$) identifies those components of the subspace $\B$ that are influenced by the covariates $x$, and its rowspace (in $\R^q$) specifies those components of (a potentially large number of) the covariates $x$ that influence $y$. Thus, \emph{the projection of the covariates $x$ onto the rowspace of $\A$ represents the interpretable component of the latent variables $\zeta$.} The term $\B_u \zeta_u$ in~\eqref{eqn:mid} represents, in some sense, the effects of those phenomena that continue to remain unobserved despite the incorporation of the covariates $x$.

Motivated by this discussion, we fit observations of $(y,x) \in \R^p \times\R^q$ to the following \emph{composite factor model} that incorporates the effects of the covariates $x$ as well as of additional unobserved latent phenomena on $y$:
\begin{equation}
y = \A x + \B_u \zeta_u + \bar{\epsilon}
\label{eqn:composite}
\end{equation}
where $\A \in \R^{p \times q}$ with $\mathrm{rank}(\A) \ll \min\{p,q\}$, $\B_u \in \R^{p \times k_u}$ with $k_u \ll p$, and the variables $\zeta_u,\bar{\epsilon}$ are independent of each other (and of $x$) and normally distributed as $\zeta_u \sim \mathcal{N}(0,\Sigma_{\zeta_u}), \bar{\epsilon} \sim \mathcal{N}(0,\Sigma_{\bar{\epsilon}})$, with $\Sigma_{\zeta_u} \succ 0, \Sigma_{\bar{\epsilon}} \succ 0$ and $\Sigma_{\bar{\epsilon}}$ being a diagonal matrix. The matrix $\A$ may also be viewed as the map specifying the best linear estimate of $y$ based on $x$. In other words, the goal is to identify a low-rank matrix $\A$ such that the conditional distribution of $y | x$ (and equivalently of $y | \A x$) is specified by a standard factor model of the form~\eqref{eqn:factormodel}.

\subsection{Composite Factor Modeling via Convex Optimization}
\label{section:compositeF}

Next we describe techniques to fit observations of $y \in \R^p$ to the model \eqref{eqn:composite}. This method is a key subroutine in our algorithmic approach for associating semantics to the latent variables in a factor model (see Section 1.3 for a high-level discussion of our approach and Section~\ref{section:experiments} for a more detailed experimental demonstration). Fitting observations of $(y,x) \in \R^p \times \R^q$ to the composite factor model \eqref{eqn:composite} is accomplished by identifying a Gaussian model over $(y,x)$ with the covariance matrix of the model satisfying certain algebraic properties.  For background on multivariate Gaussian statistical models, we refer the reader to \cite{Kay}.


The covariance matrix of $y$ in the factor model is decomposable as the sum of a low-rank matrix $\B \Sigma_\zeta \B'$ (corresponding to the $k \ll p$ latent variables $\zeta$) and a diagonal matrix $\Sigma_\epsilon$. Based on this algebraic structure, a natural approach to factor modeling is to find the smallest rank (positive semidefinite) matrix such that the difference between this matrix and the empirical covariance of the observations of $y$ is close to being a diagonal matrix (according to some measure of closeness, such as in the Frobenius norm). This problem is computationally intractable to solve in general due to the rank minimization objective \cite{Nataranjan}. As a result, a common heuristic is to replace the matrix rank by the trace functional, which results in the minimum trace factor analysis problem \cite{Ledermann,Shapiro1,S2,S3}; this problem is convex and it can be solved efficiently. The use of the trace of a positive semidefinite matrix as a surrogate for the matrix rank goes back many decades, and this topic has received much renewed interest over the past several years \cite{Meshabi,Fazel,Recht,Candes}.

In attempting to generalize the minimum-trace factor analysis approach to the composite factor model, one encounters a difficulty that arises due to the parametrization of the underlying Gaussian model in terms of covariance matrices. Specifically, with the additional covariates $x \in \R^q$ in the composite model \eqref{eqn:composite}, our objective is to identify a Gaussian model over $(y,x) \in \R^p \times \R^q$ with the joint covariance $\Sigma = \begin{pmatrix} \Sigma_y ~ \Sigma_{yx} \\ \Sigma'_{yx} ~ \Sigma_x \end{pmatrix} \in \mathbb{S}^{p + q}$ satisfying certain structural properties. One of these properties is that the conditional distribution of $y | x$ is specified by a factor model, which implies that the conditional covariance of $y | x$ must be decomposable as the sum of a low-rank matrix and a diagonal matrix. However, this conditional covariance is given by the Schur complement $\Sigma_y - \Sigma_{yx} \Sigma_x^{-1} \Sigma'_{yx}$, and specifying a constraint on the conditional covariance matrix in terms of the joint covariance matrix $\Sigma$ presents an obstacle to obtaining computationally tractable optimization formulations.

A more convenient approach to parameterizing conditional distributions in Gaussian models is to consider models specified in terms of inverse covariance matrices, which are also called \emph{precision matrices}. Specifically, the algebraic properties that we desire in the joint covariance matrix $\Sigma$ of $(y,x)$ in a composite factor model can also be stated in terms of the joint precision matrix $\Theta = \Sigma^{-1}$ via conditions on the submatrices of $\Theta = \begin{pmatrix} \Theta_y ~ \Theta_{yx} \\ \Theta'_{yx} ~ \Theta_x \end{pmatrix}$. First, the precision matrix of the conditional distribution of $y | x$ is specified by the submatrix $\Theta_y$; as the covariance matrix of the conditional distribution of $y | x$ is the sum of a diagonal matrix and a low-rank matrix, the Woodbury matrix identity implies that the submatrix $\Theta_y$ is the difference of a diagonal matrix and a low-rank matrix. Second, the rank of the submatrix $\Theta_{yx} \in \R^{p \times q}$ is equal to the rank of $\A \in \R^{p \times q}$ in non-degenerate models (i.e., if $\Sigma \succ 0$) because the relation between $\A$ and $\Theta$ is given by $\A = -[\Theta_y]^{-1} \Theta_{yx}$. Based on this algebraic structure desired in $\Theta$, we propose the following natural convex relaxation for fitting a collection of observations $\mathcal{D}^{+}_{n} = \{(y^{(i)},x^{(i)})\} _{i = 1}^{n}\subset \R^{p + q}$ to the composite model~\eqref{eqn:composite}:
\begin{eqnarray}
(\hat{\Theta}, \hat{D}_y, \hat{L}_y) = \arg\min_{\substack{\Theta \in \Sp^{p+q}, ~\Theta \succ 0 \\ D_y,L_y \in \Sp^p}} & -\ell(\Theta; \mathcal{D}_{+}^{n}) + \lambda_n [\gamma\|\Theta_{yx}\|_{\star} + \mathrm{trace}(L_y)] \nonumber \\ \mathrm{s.t.} & \Theta_{y} = D_y - L_y, ~ L_y \succeq 0, D_y ~\mathrm{is~diagonal} &
\label{eqn:main}
\end{eqnarray}
The term $\ell(\Theta; \mathcal{D}_{+}^{n})$ is the Gaussian log-likelihood function that enforces fidelity to the data, and it is given as follows (up to some additive and multiplicative terms):
\begin{equation}
\ell(\Theta;\mathcal{D}_{+}^{n}) = \log\det(\Theta) - \mathrm{trace}\left[\Theta \cdot \tfrac{1}{n}\sum_{i=1}^{n} \begin{pmatrix} y^{(i)} \\ x^{(i)} \end{pmatrix} \begin{pmatrix} y^{(i)} \\ x^{(i)} \end{pmatrix}' \right].
\end{equation}
This function is concave as a function of the joint precision matrix\footnote{An additional virtue of parameterizing our problem in terms of precision matrices rather than in terms of covariance matrices is that the log-likelihood function in Gaussian models is not concave over the cone of positive semidefinite matrices when viewed as a function of the covariance matrix.} $\Theta$. The matrices $D_y, L_y$ represent the diagonal and low-rank components of $\Theta_y$. As with the idea behind minimum-trace factor analysis, the role of the trace norm penalty on $L_y$ is to induce low-rank structure in this matrix. Based on a more recent line of work originating with the thesis of Fazel \cite{Fazel,Recht,Candes}, the nuclear norm penalty $\|\Theta_{yx}\|_\star$ on the submatrix $\Theta_{yx}$ (which is in general a non-square matrix) is useful for promoting low-rank structure in that submatrix of $\Theta$. The parameter $\gamma$ provides a tradeoff between the observed/interpretable and the unobserved parts of the composite factor model \eqref{eqn:composite}, and the parameter $\lambda_n$ provides a tradeoff between the fidelity of the model to the data and the overall complexity of the model (the total number of observed and unobserved components in the composite model \eqref{eqn:composite}). In summary, for $\lambda_n, \gamma \geq 0$ the regularized maximum-likelihood problem \eqref{eqn:main} is a convex program. From the optimal solution $(\hat{\Theta}, \hat{D}_y, \hat{L}_y)$ of \eqref{eqn:main}, we can obtain estimates for the parameters of the composite factor model \eqref{eqn:composite} as follows:
\begin{equation}
\begin{aligned}
\hat{\A} &= -[\hat{\Theta}_y]^{-1} \hat{\Theta}_{yx} \\ \hat{\B}_u &= \mathrm{any~squareroot~of~}(\hat{D}_y - \hat{L}_y)^{-1} - \hat{D}_y^{-1}~\mathrm{such~that}~ \hat{\B}_u \in \R^{p \times \rk(\hat{L}_y)},
\end{aligned}
\end{equation}
with the covariance of $\zeta_u$ being the identity matrix of appropriate dimensions and the covariance of $\bar{\epsilon}$ being $\hat{D}_y^{-1}$. 
The convex program \eqref{eqn:main} is log-determinant semidefinite programs that can be solved efficiently using existing numerical solvers such as the LogDetPPA package \cite{Toh}.

\subsection{{{}Algorithmic Approach for Interpreting Latent Variables in a Factor Model}}
\label{section:algorithm}

Our discussion has led us to a natural (meta-) procedure for interpreting latent variables in a factor model. Suppose that we are given a factor model underlying $y \in \mathbb{R}^p$. The analyst proceeds by obtaining simultaneous measurements of the variables $y$ as well as some additional covariates $x \in \R^q$ of plausibly relevant phenomena. Based on these joint observations, we identify a suitable composite factor model \eqref{eqn:composite} via the convex program \eqref{eqn:main}.  In particular, we sweep over the parameters $\lambda_n,\gamma$ in \eqref{eqn:main} to identify composite models that achieve a suitable decomposition -- in terms of effects attributable to the additional covariates $x$ and of effects corresponding to remaining unobserved phenomena -- of the effects of the latent variables in the factor model given as input.\\

To make this approach more formal, consider a composite factor model~\eqref{eqn:composite} $y = \mathcal{A}{x} + \mathcal{B}_u\zeta_{u} + \epsilon$ underlying a pair of random vectors $(y,x) \in \R^{p} \times \R^q$, with $\rk(\A) = k_x$, $\B_u \in \R^{p \times k_u}$, and $\colspace({\A}) \cap \colspace({\B_u}) = \{0\}$. As described in Section~\ref{section:compositeF}, the algebraic aspects of the underlying composite factor model translate to algebraic properties of submatrices of $\Theta \in \mathbb{S}^{p+q}$. In particular, the submatrix $\Theta_{yx}$ has rank equal to $k_x$ and the submatrix $\Theta_{y}$ is decomposable as $D_y - L_y$ with $D_y$ being diagonal and $L_y \succeq 0$ having rank equal to $k_u$. Finally, the transversality of $\colspace({\A})$ and $\colspace({\B_u})$ translates to the fact that $\colspace(\Theta_{yx}) \cap \colspace(L_y) = \{0\}$ have a transverse intersection. One can simply check that the factor model underlying the random vector $y \in \R^p$ that is induced upon marginalization of $x$ is specified by the precision matrix of $y$ given by $\tilde{\Theta}_y = D_y - [L_y + \Theta_{yx}(\Theta_x)^{-1}\Theta_{xy}]$. Here, the matrix $L_y + \Theta_{yx}(\Theta_x)^{-1}\Theta_{xy}$ is a rank $k_x + k_u$ matrix that captures the effect of latent variables in the factor model. This effect is decomposed into $\Theta_{yx}(\Theta_x)^{-1}\Theta_{xy}$ --  a rank $k_x$ matrix representing the component of this effect attributed to $x$, and $L_y$ -- a matrix of rank $k_u$ representing the effect attributed to residual latent variables. \\

These observations motivate the following algorithmic procedure. Suppose we are given a factor model that specifies the precision matrix of $y$ as the difference $\hat{\tilde{D}}_{y} - \hat{\tilde{L}}_{y}$, where $\hat{\tilde{D}}_{y}$ is diagonal and $\hat{\tilde{L}}_{y}$ is low rank. Then the composite factor model of $(y,x)$ with estimates $(\hat{\Theta}, \hat{D}_y, \hat{L}_y)$ offers an interpretation of the latent variables of the given factor model if $(i)~ \mathrm{rank}(\hat{\tilde{L}}_y) = \mathrm{rank}(\hat{L}_y + \hat{\Theta}_{yx}\hat{\Theta}_x^{-1}\hat{\Theta}_{xy})$, $(ii)~\colspace(\hat{\Theta}_{yx}) \cap \colspace(\hat{L}_{y}) = \{0\}$, and \\$(iii) \max\{\|\hat{\tilde{D}}_y - \hat{D}_y \|_2 /\|\hat{\tilde{D}}_y\|_2, \|\hat{\tilde{L}}_y - [\hat{L}_y+\hat{\Theta}_{yx}\hat{\Theta}_x^{-1}\hat{\Theta}_{xy}]\|_2 / \|\hat{\tilde{L}}_y\|_2\}$ is small. The full algorithmic procedure for attributing meaning to latent variables of a factor model is outlined below:
\FloatBarrier
\begin{algorithm}
\caption{Interpreting Latent Variables in a Factor Model}
\begin{algorithmic}[1]
\STATE {\bf Input}: A collection of observations $\mathcal{D}^{+}_{n} = \{(y^{(i)}, x^{(i))}\}_{i=1}^{n} \subset \R^p \times \R^q$ of the variables $y$ and of some auxiliary covariates $x$; Factor model with parameters $(\hat{\tilde{D}}_y, \hat{\tilde{L}}_y)$. \\
\vspace{.04in}
\STATE {\bf Composite Factor Modeling}: For each $d = 1, \dots, q$, sweep over parameters $(\lambda_n, \gamma)$ in the convex program \eqref{eqn:main} (with $\mathcal{D}_n^+$ as input) to identify composite models with estimates $(\hat{\Theta},\hat{D}_y,\hat{L}_y)$ that satisfy the following three properties: $(i)~\text{rank}(\hat{\Theta}_{yx}) = d$, $(ii)$~$\text{rank}(\hat{\tilde{L}}_y) = \text{rank}(\hat{L}_y) + \text{rank}(\hat{\Theta}_{yx})$, and $(iii)~\text{rank}(\hat{\tilde{L}}_y) = \text{rank}(\hat{L}_y + \text{rank}(\hat{\Theta}_{yx}\hat{\Theta}_{x}^{-1}\hat{\Theta}_{xy}))$. \\
\vspace{.04in}
\STATE{\bf Identifying Subspace}: For each $d = 1, \dots, q$ and among the candidate composite models (from the previous step), choose the composite factor model that minimizes the quantity $\max\{\|\hat{\tilde{D}}_y - \hat{D}_y \|_2 /\|\hat{\tilde{D}}_y\|_2, \|\hat{\tilde{L}}_y - [\hat{L}_y+\hat{\Theta}_{yx}\hat{\Theta}_x^{-1}\hat{\Theta}_{xy}]\|_2 / \|\hat{\tilde{L}}_y\|_2\}$.\\
\vspace{.04in}
\STATE{\bf Output}: For each $d = 1,\dots q$, the $d$-dimensional projection of $x$ into the row-space of $\hat{\Theta}_{yx}$ represents the interpretable component of the latent variables in the factor model.
\end{algorithmic}
\end{algorithm}
\FloatBarrier

 The effectiveness of Algorithm 1 is dependent on the size of the quantity $\max\{\|\hat{\tilde{D}}_y - \hat{D}_y \|_2 /\|\hat{\tilde{D}}_y\|_2, \|\hat{\tilde{L}}_y - \hat{L}_y-\hat{\Theta}_{yx}\hat{\Theta}_x^{-1}\hat{\Theta}_{xy}]\|_2 / \|\hat{\tilde{L}}_y\|_2\}$. The smaller this quantity, the better the composite factor model fits to the given factor model. Finally, recall from Section~\ref{secion:composite} that the projection of covariates $x$ onto to the row-space of $\A$ (from the composite model \eqref{eqn:composite}) represents the interpretable component of the latent variables of the factor model. Because of the relation $\A = -[\Theta_y]^{-1} \Theta_{yx}$, this interpretable component is obtained by projecting the covariates $x$ onto the row-space of $\Theta_{yx}$. This observation explains the final step of Algorithm 1.  

The input to Algorithm 1 is a factor model underlying a collection of variables $y \in \R^p$, and the algorithm proceeds to obtain semantic interpretation of the latent variables of the factor model. However, in many situations, a factor model underlying $y \in \R^p$ may not be available in advance, and must be learned in a data-driven fashion based on observations of $y \in \R^p$. In our experiments (see Section~\ref{section:experiments}), we learn a factor model using a specialization of the convex program \eqref{eqn:main}. It is reasonable to ask whether one might directly fit to a composite model to the covariates and responses jointly without reference to the underlying factor model based on the responses. However, in our experience with applications, it is often the case that observations of the responses $y$ are much more plentiful than of joint observations of responses $y$ and covariates $x$. As an example, consider a setting in which the responses are a collection of financial asset prices (such as stock return values); observations of these variables are available at a very fine time-resolution on the order of seconds. On the other hand, some potentially useful covariates such as GDP, government expenditures, federal debt, and consumer rate are available at a much coarser scale (usually on the order of months or quarters). As another example, consider a setting in which the responses are reservoir volumes of California; observations of these variables are available at a daily scale. On the other hand, reasonable covariates that one may wish to associate to the latent variables underlying California reservoir volumes such as agricultural production, crop yield rate, average income, and population growth rate are available at a much coarser time scale (e.g. monthly). In such settings, the analyst can utilize the more abundant set of observations of the responses $y$ to learn an accurate factor model first. Subsequently, one can employ our approach to associate semantics to the latent variables in this factor model based on the potentially limited number of observations of the responses $y$ and the covariates $x$. 

\subsection{{Our Results}}
\label{section:results}

In Section~\ref{section:theorem} we carry out a theoretical analysis to investigate whether the framework outlined in Algorithm 1 can succeed. We discuss a model problem setup, which serves as the basis for the main theoretical result in Section~\ref{section:theorem}. Suppose we have Gaussian random vectors $(y,x) \in \R^p \times \R^q$ that are related to each other via a composite factor model~\eqref{eqn:composite}.  Note that this composite factor model induces a factor model underlying the variables $y \in \R^p$ upon marginalization of the covariates $x$. In the subsequent discussion, we assume that the factor model that is supplied as input to Algorithm 1 is the factor model underlying the responses $y$. \\
 Now we consider the following question: Given observations jointly of $(y,x) \in \R^{p+q}$, does the convex relaxation \eqref{eqn:main} (for suitable choices of regularization parameters $\lambda_n,\gamma$) estimate the composite factor model underlying these two random vectors accurately? An affirmative answer to this question demonstrates the success of Algorithm 1. In particular, a positive answer to this question implies that we can decompose the effects of the latent variables in the factor model underlying $y$ using the convex relaxation \eqref{eqn:main}, as the accurate estimation of the composite model underlying $(y,x)$ implies a successful decomposition of the effects of the latent variables in the factor model underlying $y$.  That is, steps 2-3 in the Algorithm are successful. In Section~\ref{section:theorem}, we show that under suitable identifiability conditions on the population model of the joint random vector $(y,x)$, the convex program \eqref{eqn:main} succeeds in solving this question. Our analysis is carried out in a high-dimensional asymptotic scaling regime in which the dimensions $p,q$, the number of observations $n$, and other model parameters may all grow simultaneously \cite{Buhlmann,Wainwright}. \\

We give concrete demonstration of Algorithm 1 with experiments on synthetic data and real-world financial data. For the financial asset problem, we consider as our variables $y$ the monthly averaged stock prices of $45$ companies from the Standard and Poor index over the period March 1982 to March 2016, and we identify a factor model \eqref{eqn:factormodel} over $y$ with $10$ latent variables (the approach we use to fit a factor model is described in Section~\ref{section:experiments}). We then obtain observations of $q = 13$ covariates on quantities related to oil trade, GDP, government expenditures, etc. (See Section~\ref{section:experiments} for the full list), as these plausibly influence stock returns. Following the steps outlined in Algorithm 1, we use the convex program \eqref{eqn:main} to identify a two-dimensional projection of these $13$ covariates that represent an interpretable component of the $10$ latent variables in the factor model, as well as a remaining set of $8$ latent variables that constitute phenomena not observed via the covariates $x$. In further analyzing the characteristics of the two-dimensional projection, we find that EUR to USD exchange rate and government expenditures are the most relevant of the $13$ covariates considered in our experiment, while mortgage rate and oil imports are less useful. See Section~\ref{section:experiments} for complete details.

\subsection{Related Work}
Elements of our approach bear some similarity with \emph{canonical correlations analysis} \cite{Hotelling}, which is a classical technique for identifying relationships between two sets of variables. In particular, for a pair of jointly Gaussian random vectors $(y,x) \in \R^{p \times q}$, canonical correlations analysis may be used as a technique for identifying the most relevant component(s) of $x$ that influence $y$. However, the composite factor model \eqref{eqn:composite} allows for the effect of further unobserved phenomena not captured via observations of the covariates $x$. Consequently, our approach in some sense incorporates elements of both canonical correlations analysis and factor analysis. It is important to note that algorithms for factor analysis and for canonical correlations analysis usually operate on covariance and cross-covariance matrices. However, we parametrize our regularized maximum-likelihood problem \eqref{eqn:main} in terms of precision matrices, which is a crucial ingredient in leading to a computationally tractable convex program.

The nuclear-norm heuristic has been employed widely over the past several years in a range of statistical modeling tasks involving rank minimization problems; see \cite{Wainwright} and the references therein. The proof of our main result in Section~\ref{section:theorem} incorporates some elements from the theoretical analyses in these previous papers, along with the introduction of some new ingredients. We give specific pointers to the relevant literature in Section~\ref{section:proofs}.

\subsection{Notation}Given a matrix $U \in \mathbb{R}^{p_1 \times p_2}$, and the norm $\|U\|_2$ denotes the spectral norm (the largest singular value of $U$). We define the linear operators $\mathcal{F}: \Sp^p \times \Sp^p \times \mathbb{R}^{p{\times}q} \times \Sp^q \rightarrow \Sp^{(p+q)}$ and its adjoint $\mathcal{F}^{\dagger}: \Sp^{(p+q)} \rightarrow \Sp^p \times \Sp^p \times \mathbb{R}^{p{\times}q} \times \Sp^q$ as follows:
\begin{equation}
\mathcal{F}(M, N, K, O) \triangleq \left( \begin{array}{cc}
M - N & K \\
K^T & O \end{array} \right), \qquad \mathcal{F}^{\dagger}\left( \begin{array}{cc}
Q & K \\
K^{T} & O \end{array} \right)\triangleq (Q,Q,K,O)
\label{eqn:OperatorDefs}
\end{equation}
Similarly, we define the linear operators 
$\mathcal{G}: \Sp^p \times \mathbb{R}^{p{\times}q} \rightarrow \Sp^{(p+q)}$ and its adjoint $\mathcal{G}^{\dagger}: \Sp^{(p+q)} \rightarrow \Sp^p \times \mathbb{R}^{p{\times}q}$ as follows:
\begin{equation}
\mathcal{G}(M, K) \triangleq \left( \begin{array}{cc}
M  & K \\
K^T & 0 \end{array} \right), \qquad \mathcal{G}^{\dagger}\left( \begin{array}{cc}
Q & K \\
K^{T} & O \end{array} \right)\triangleq (Q,K)
\label{eqn:OperatorDefs23}
\end{equation}
Finally, for any subspace $\s$, the projection onto the subspace is denoted by $\mathcal{P}_\s$.

%% file: TheoreticalResults.tex
\section{{Theoretical Results}}
\label{section:theorem}
In this section, we state a theorem to prove the consistency of convex program \eqref{eqn:main}. This theorem requires assumptions on the population precision matrix, which are discussed in Section~\ref{section:Fishercond}. We provide examples of population composite factor models \eqref{eqn:main} that satisfy these conditions. The theorem statement is given in Section~\ref{section:theoremstatement} and the proof of the theorem is given in Section~\ref{section:proofs} with some details deferred to the appendix.  
\subsection{{Technical Setup}}
\label{section:setup}
As discussed in Section~\ref{section:results}, our theorems are premised on the existence of a population composite factor model \eqref{eqn:composite} $y = \A^\star{x} + \B^\star_u \zeta_u + \epsilon$ underlying a pair of random vectors $(y,x) \in \R^p \times \R^q$, with $\rk(\A^\star) = k_x$, $\B_u^\star \in \R^{p \times k_u}$, and $\colspace({\A}^\star) \cap \colspace({\B_u}^\star) = \{0\}$. As the convex relaxation \eqref{eqn:main} is solved in the precision matrix parametrization, the conditions for our theorems are more naturally stated in terms of the joint precision matrix $\Theta^\star \in \mathbb{S}^{p+q}, ~ \Theta^\star \succ 0$ of $(y,x)$.  The algebraic aspects of the parameters underlying the factor model translate to algebraic properties of submatrices of $\Theta^\star$.  In particular, the submatrix $\Theta^\star_{yx}$ has rank equal to $k_x$, and the submatrix $\Theta^\star_y$ is decomposable as $D_y^\star - L_y^\star$ with $D^\star_y$ being diagonal and $L^\star_y \succeq 0$ having rank equal to $k_u$.  Finally, the transversality of $\colspace({\A}^\star)$ and $\colspace({\B_u}^\star)$ translates to the fact that $\colspace(\Theta^\star_{yx}) \cap \colspace(L_y^\star) = \{0\}$ have a transverse intersection.

To address the requirements raised in Section~\ref{section:results}, we seek an estimate $(\hat{\Theta}, \hat{D}_y, \hat{L}_y)$ from the convex relaxation \eqref{eqn:main} such that $\mathrm{rank}(\hat{\Theta}_{yx}) = \mathrm{rank}(\Theta^\star_{yx})$, $\mathrm{rank}(\hat{L}_y) = \mathrm{rank}(L^\star_y),$ and that $\|\hat{\Theta}-\Theta^\star\|_2$ is small. Building on both classical statistical estimation theory \cite{Bickel} as well as the recent literature on high-dimensional statistical inference \cite{Buhlmann,Wainwright}, a natural set of conditions for obtaining accurate parameter estimates is to assume that the curvature of the likelihood function at $\Theta^\star$ is bounded in certain directions. This curvature is governed by the Fisher information at $\Theta^\star$:
\begin{eqnarray*}
\mathbb{I}^\star \triangleq {{{{\Theta}}}^\star}^{-1} \otimes {{\Theta}^\star}^{-1} = \Sigma^\star \otimes \Sigma^\star.
\end{eqnarray*}
Here $ \otimes$ denotes a tensor product between matrices and $\mathbb{I}^\star$ may be viewed as a map from $\mathbb{S}^{(p+q)}$ to $\mathbb{S}^{(p+q)}$. We impose conditions requiring that $\mathbb{I}^\star$ is well-behaved when applied to matrices of the form \\ $\Theta - \Theta^\star = \begin{pmatrix} (D_y-D_y^\star)-(L_y-L_y^\star) & \Theta_{yx} - \Theta_{yx}^\star \\ {\Theta_{yx}}'-{\Theta_{yx}^\star}' & \Theta_x-\Theta_{x}^\star \end{pmatrix}$, where $(L_y,\Theta_{yx})$ are in a neighborhood of $(L_y^\star,\Theta_{yx}^\star)$ restricted to sets of low-rank matrices. These local properties of $\mathbb{I}^\star$ around $\Theta^\star$ are conveniently stated in terms of \emph{tangent spaces} to the algebraic varieties of low-rank matrices. In particular, the tangent space at a rank-$r$ matrix $N \in \R^{p_1 \times p_2}$ with respect to the algebraic variety of $p_1 \times p_2$ matrices with rank less than or equal to $r$ is given by\footnote{We also consider the tangent space at a symmetric low-rank matrix with respect to the algebraic variety of symmetric low-rank matrices. We use the same notation `$T$' to denote tangent spaces in both the symmetric and non-symmetric cases, and the appropriate tangent space is clear from the context.}:
\begin{eqnarray*}
T(N) &\triangleq& \{N_R + N_C | N_R, N_C \in \R^{p_1 \times p_2}, \\&~& \rowspace~({N_R}) \subseteq \rowspace~({N}), \\ & & \colspace~({N_C}) \subseteq \colspace~ ({N})\}
\end{eqnarray*}

In the next section, we describe conditions on the population Fisher information $\mathbb{I}^\star$ in terms of the tangent spaces $T(L_{y}^\star)$, and $T(\Theta_{yx}^\star)$; under these conditions, we present a theorem in Section \ref{section:theoremstatement} showing that the convex program \eqref{eqn:main} obtains accurate estimates.

\subsection{{Fisher Information Conditions}}
\label{section:Fishercond}
Given a norm $\|\cdot\|_{\Upsilon}$ on $\mathbb{S}^p \times \mathbb{S}^p \times \mathbb{R}^{p \times q} \times \mathbb{S}^{q}$, we first consider a classical condition in statistical estimation literature, which is to control the minimum gain of the Fisher information $\mathbb{I}^\star$ restricted to a subspace $\mathbb{H} \subset \mathbb{S}^p \times \mathbb{S}^p \times \mathbb{R}^{p \times q} \times \mathbb{S}^{q}$ as follows:
\begin{eqnarray}
\chi({\mathbb{H}}, \|\cdot\|_{\Upsilon}) \triangleq \min_{\substack{Z \in {\mathbb{H}}\\ \|Z\|_\Upsilon= 1}} \|\mathcal{P}_{\mathbb{H}} \F^{\dagger} \mathbb{I}^\star \F \mathcal{P}_{\mathbb{H}}(Z)\|_{\Upsilon},
\label{eqn:Chieq}
\end{eqnarray}
where $\mathcal{P}_{\mathbb{H}}$ denotes the projection operator onto the subspace $\mathbb{H}$ and the linear maps $\F$ and $\F^{\dagger}$ are defined in \eqref{eqn:OperatorDefs}. The quantity $\chi({\mathbb{H}}, \|\cdot\|_{\Upsilon})$ being large ensures that the Fisher information $\mathbb{I}^\star$ is well-conditioned restricted to image $\F \mathbb{H} \subseteq \mathbb{S}^{p+q}$. The remaining conditions that we impose on $\mathbb{I}^\star$ are in the spirit of irrepresentibility-type conditions \cite{Meinhausen,Zao,Wai2009,Ravikumar,Chand2012} that are frequently employed in high-dimensional estimation. In the subsequent discussion, we employ the following notation to denote restrictions of a subspace $\mathbb{H} = H_1 \times H_2 \times H_3 \times H_4  \subset \mathbb{S}^{ p} \times \mathbb{S}^{p} \times \R^{p \times q} \times \mathbb{S}^{q}$ (here $H_1,H_2,H_3,H_4$ are subspaces in $\mathbb{S}^{ p},\mathbb{S}^p,\R^{p \times q},\mathbb{S}^q$, respectively) to its individual components. The restriction to the second components of $\mathbb{H}$ is given by $\mathbb{H}[2] = H_2$. The restriction to the second and third component of $\mathbb{H}$ is given by $\mathbb{H}[2,3] = H_2 \times H_3 \subset \mathbb{S}^{ p} \times \R^{p \times q}$. Given a norm $\|.\|_{\Pi}$ on $\mathbb{S}^p \times \mathbb{R}^{p \times q}$, we control the gain of $\mathbb{I}^\star$ restricted to $\mathbb{H}[2,3]$
\begin{equation}
\begin{aligned}
\Xi(\mathbb{H}, \|\cdot\|_{\Pi}) \triangleq \min_{\substack{Z \in \mathbb{H}[2,3]\\ \|Z\|_{\Pi} = 1}} \|\mathcal{P}_{{{\mathbb{H}[2,3]}}} \mathcal{G}^{\dagger}\mathbb{I}^{\star}\mathcal{G}\mathcal{P}_{{{\mathbb{H}[2,3]}}}(Z)\|_{\Pi}.
\end{aligned}
\label{eqn:Xidef}
\end{equation}
Here, the linear maps $\mathcal{G}$ and $\mathcal{G}^\dagger$ are defined in \eqref{eqn:OperatorDefs23}. In the spirit of irrepresentability conditions, we control the inner-product between elements in $\mathcal{G} \mathbb{H}[2,3]$ and $\mathcal{G} \mathbb{H}[2,3]^{\perp}$, as quantified by the metric induced by $\mathbb{I}^\star$ via the following quantity
\begin{equation}
\begin{aligned}
\varphi(\mathbb{H}, \|\cdot\|_{\Pi}) \triangleq \max_{\substack{Z \in \mathbb{H}[2,3]\\ \|Z\|_{\Pi} = 1}} \|\mathcal{P}_{{{\mathbb{H}[2,3]^{\perp}}}}\mathcal{G}^{\dagger} \mathbb{I}^{\star}\mathcal{G}\mathcal{P}_{{{\mathbb{H}}[2,3]}} (\mathcal{P}_{{{\mathbb{H}[2,3]}}} \mathcal{G}^{\dagger}\mathbb{I}^{\star}\mathcal{G}\mathcal{P}_{{{\mathbb{H}[2,3]}}})^{-1}(Z)\|_{\Pi}.
\end{aligned}
\label{eqn:varphidef}
\end{equation}
The operator $(\mathcal{P}_{{{\mathbb{H}[2,3]}}} \mathcal{G}^{\dagger}\mathbb{I}^{\star}\mathcal{G}\mathcal{P}_{{{\mathbb{H}[2,3]}}})^{-1}$ in \eqref{eqn:varphidef} is well-defined if $\Xi({\mathbb{H}}) > 0$, since this latter condition implies that $\mathbb{I}^\star$ is injective restricted to $\mathcal{G} \mathbb{H}[2,3]$. The quantity $\varphi(\mathbb{H}, \|\cdot\|_{\Upsilon})$ being small implies that any element of $\mathcal{G} \mathbb{H}[2,3]$ and any element of $\mathcal{G}{\mathbb{H}[2,3]^{\perp}}$ have a small inner-product (in the metric induced by $\mathbb{I}^\star$). The reason that we restrict this inner product to the second and third components of $\mathbb{H}$ in the  quantity $\varphi(\mathbb{H}, \|.\|_{\Upsilon})$ is that the regularization terms in the convex program \eqref{eqn:main} are only applied to the matrices $L_y$ and $\Theta_{yx}$. 

A natural approach to controlling the conditioning of the Fisher information around $\Theta^\star$ is to bound the quantities $\chi(\mathbb{H}^\star,\|\cdot\|_\Upsilon)$,  $\Xi(\mathbb{H}^\star, \|\cdot\|_{\Pi})$, and $\varphi(\mathbb{H}^\star, \|\cdot\|_{\Upsilon})$ for $\mathbb{H}^\star = \mathcal{W} \times T(L_y^\star) \times T(\Theta_{yx}^\star) \times \Sp^q$ where $\mathcal{W} \in \Sp^p$ is the set of diagonal matrices. However, a complication that arises with this approach is that the varieties of low-rank matrices are locally curved around $L_y^\star$ and around $\Theta_{yx}^\star$. Consequently, the tangent spaces at points in neighborhoods around $L_y^\star$ and around $\Theta_{yx}^\star$ are not the same as $T(L_y^\star)$ and $T(\Theta_{yx}^\star)$. In order to account for this curvature underlying the varieties of low-rank matrices, we bound the distance between nearby tangent spaces via the following induced norm:
\begin{eqnarray*}
\rho(T_1,T_2) \triangleq \max_{\|N\|_2 \leq 1} \|(\mathcal{P}_{T_1} - \mathcal{P}_{T_2})(N)\|_2.
\label{eqn:distortion}
\end{eqnarray*}
The quantity $\rho(T_1,T_2)$ measures the largest angle between $T_1$ and $T_2$.  Using this approach for bounding nearby tangent spaces, we consider subspaces $\mathbb{H}' = \mathcal{W} \times T'_y \times T'_{yx} \times \mathbb{S}^{q}$ for all $T_y'$ close to $T(L_y^\star)$ and for all $T_{yx}'$ close to $T(\Theta_{yx}^\star)$, as measured by $\rho$ \cite{Chand2012}. For $\omega_y \in (0,1)$ and $\omega_{yx} \in (0,1)$, we bound $\chi({\mathbb{H}}',  \|\cdot\|_{\Upsilon})$, $\Xi(\mathbb{H}', \|\cdot\|_{\Pi})$, and $\varphi({\mathbb{H}}', \|\cdot\|_{\Pi})$ in the sequel for all subspaces $\mathbb{H}'$ in the following set:
\begin{equation}
\begin{aligned}
U{(\omega_y, \omega_{yx})} \triangleq \Big\{\mathcal{W} \times T'_y \times T'_{yx} \times \mathbb{S}^{q} ~|~ &\rho({{T_{y}'}}, T({L_{y}^\star})) \leq \omega_y \\ & \rho({{T_{yx}'}}, T({\Theta_{yx}^\star})) \leq \omega_{yx}\Big\}.
\end{aligned}
\label{eqn:Udef}
\end{equation}

We control the quantities $\Xi({\mathbb{H}'}, \|\cdot\|_{\Pi})$ and $\varphi({\mathbb{H}}', \|\cdot\|_{\Pi})$ using the dual norm of the regularizer $\mathrm{trace}(L_y) + \gamma \|\Theta_{yx}\|_\star$ in \eqref{eqn:main}:
\begin{equation}
\Gamma_{\gamma}(L_y, \Theta_{yx}) \triangleq \max\left\{\|L_y\|_{2}, \frac{\|\Theta_{yx}\|_{2}}{\gamma}\right\}.
\label{eqn:Gammadef}
\end{equation}
Furthermore, we control the quantity $\chi({\mathbb{H}}',  \|\cdot\|_{\Upsilon})$ using a slight variant of the dual norm:
\begin{equation}
\Phi_{\gamma}(D_y, L_y, \Theta_{yx}, \Theta_{x}) \triangleq \max\left\{\|D_y\|_2,\|L_y\|_{2}, \frac{\|\Theta_{yx}\|_{2}}{\gamma}, \|\Theta_{x}\|_2 \right\}.
\label{eqn:Phidef}
\end{equation}
As the dual norm $\max\left\{\|L_y\|_{2}, \frac{\|\Theta_{yx}\|_{2}}{\gamma}\right\}$ of the regularizer in \eqref{eqn:main} plays a central role in the optimality conditions of \eqref{eqn:main}, controlling the quantities $\chi({\mathbb{H}'},\Phi_{\gamma})$, $\Xi({\mathbb{H}}', \Gamma_{\gamma})$, and $\varphi({\mathbb{H}}', \Gamma_{\gamma})$ leads to a natural set of conditions that guarantee the consistency of the estimates produced by
\eqref{eqn:main}. In summary, given a fixed set of parameters $(\gamma,\omega_y, \omega_{yx}) \in \R_+ \times (0,1) \times (0,1)$, we assume that $\mathbb{I}^\star$ satisfies the following conditions:
\begin{eqnarray}
\mathrm{Assumption~1}&:& \inf_{\mathbb{H}' \in U{(\omega_y, \omega_{yx})}}\chi({\mathbb{H}}', {{\Phi}_{\gamma}}) \geq \alpha, ~~~ \mathrm{for~some~} \alpha > 0 \label{eqn:FirstFisherCond} \\[.1in]
\mathrm{Assumption~2}&:& \inf_{\mathbb{H}' \in U{(\omega_y, \omega_{yx})}}\Xi({\mathbb{H}}', {{\Gamma}_{\gamma}}) > 0  \label{eqn:FirstFisherCond3} \\[.1in]
\mathrm{Assumption~3}&:& \sup_{\mathbb{H}' \in U{(\omega_{yx}ß, \omega_{yx})}}\varphi({\mathbb{H}}', {{\Gamma}_{\gamma}}) \leq 1-\frac{2}{\beta+1} ~~~ \mathrm{for~some~} \beta \geq 2.
\label{eqn:SecondFisherCond}
\end{eqnarray}
For fixed $(\gamma, \allowbreak \omega_y, \omega_{yx})$, larger value of $\alpha$ and smaller value of $\beta$ in these assumptions lead to a better conditioned $\mathbb{I}^\star$.

Assumptions 1, 2, and 3 are analogous to conditions that play an important role in the analysis of the Lasso for sparse linear regression, graphical model selection via the Graphical Lasso \cite{Ravikumar}, and in several other approaches for high-dimensional estimation. As a point of comparison with respect to analyses of the Lasso, the role of the Fisher information $\mathbb{I}^\star$ is played by $A^TA$, where $A$ is the underlying design
matrix. In analyses of both the Lasso and the Graphical Lasso in the papers referenced above, the analog of the subspace $\mathbb{H}$ is the set of models with support contained inside the support of the underlying sparse population model. Assumptions 1, 2, and 3 are also similar in spirit to conditions employed in the analysis of convex relaxation methods for latent-variable graphical model selection \cite{Chand2012}.
\subsection{{When Do the Fisher Information Assumptions Hold?}}
In this section, we provide examples of composite models \eqref{eqn:composite} that satisfy Assumptions 1, 2 and 3 in \eqref{eqn:FirstFisherCond} \eqref{eqn:FirstFisherCond3}, and \eqref{eqn:SecondFisherCond} for some choices of $\alpha > 0$, $\beta \geq 2$, $\omega_y \in (0,1)$,  $\omega_{yx} \in (0,1)$ and $\gamma > 0$ . Specifically, consider a population composite factor model $y = \mathcal{A}^\star{x} + \mathcal{B}_u^\star\zeta_u + \bar{\epsilon}$, where $\A^\star \in \R^{p \times q}$ with $\text{rank}(\A^\star) = k_x$, $\mathcal{B}_u^\star \in \R^{p,k_u}$, $\colspace({\A^\star}) \cap \colspace({\B_u^\star}) = \{0\}$, and the random variables $\zeta_u,\bar{\epsilon}, x$ are independent of each other and normally distributed as $\zeta_u \sim \mathcal{N}(0,\Sigma_{\zeta_u}), \bar{\epsilon} \sim \mathcal{N}(0,\Sigma_{\bar{\epsilon}})$.  
As described in Section \ref{section:compositeF}, the properties of the composite factor model translate to algebraic properties on the underlying precision matrix $\Theta^\star \in \Sp^{p+q}$. Namely, the submatrix $\Theta_{yx}^\star$ has rank equal to $k_x$ and the submatrix $\Theta_{y}^\star$ is decomposable as $D_y^\star - L_y^\star$ with $D_y^\star$ being diagonal and $L_y^\star \succeq 0$ having rank equal to $k_u$. Recall that the factor model underlying the random vector $y \in \R^p$ that is induced upon marginalization of $x$ is specified by the precision matrix of $y$ given by $\tilde{\Theta}_y^\star = D_y^\star - \Big[L_y^\star + \Theta_{yx}^\star(\Theta_x^\star)^{-1}\Theta_{xy}^\star\Big]$. Here, $L_y^\star + \Theta_{yx}^\star(\Theta_x^\star)^{-1}\Theta_{xy}^\star$ represents the effect of the latent variables in the underlying factor model. When learning a composite factor model, this effect is decomposed into: $\Theta_{yx}^\star(\Theta_x^\star)^{-1}\Theta_{xy}^\star$ --  a rank $k_x$ matrix representing the component of this affect attributed to $x$ -- and $L_y^\star$ -- a matrix of rank $k_u$ representing the effect of residual latent variables. There are two identifiability concerns that arise when learning a composite factor model. First, the low rank matrices $L_y^\star$ and $\Theta_{yx}^\star(\Theta_x^\star)^{-1}\Theta_{xy}^\star$ must be distinguishable from the diagonal matrix $D_y^\star$. Following previous literature in diagonal and low rank matrix decompositions \cite{Saunderson,Chand2012}, this task can be achieved by ensuring that the column/row spaces of $L_y^\star$ and $\Theta_{yx}^\star(\Theta_x^\star)^{-1}\Theta_{xy}^\star$ are \emph{incoherent} with respect to the standard basis. Specifically, given a subspace $U \subset \R^p$, the coherence of the subspace $U$ is defined as:
\begin{eqnarray*}
\mu(U) = \max_{i = 1,2 \dots p} \|\mathcal{P}_{U}(e_i)\|_{\ell_2}^2
\end{eqnarray*}
where $\mathcal{P}$ denotes a projection operation and $e_i \in \R^p$ denotes the i'th standard basis vector. It is not difficult to show that this incoherence parameter satisfies the following inequality:
\begin{eqnarray*}
\frac{\text{dim}(U)}{p} \leq \mu(U) \leq 1.
\end{eqnarray*}
A subspace $U$ with small coherence is necessarily of small dimension and far from containing standard basis elements. As such, a symmetric matrix with incoherent row and column spaces is low-rank and quite different from being a diagonal matrix. Consequently, we require that the quantities $\mu(\text{column-space}(L_y^\star))$ and $\mu(\text{column-space}(\Theta_{yx}^\star{\Theta_x^\star}^{-1}\Theta_{xy}^\star))$ are small \footnote{We only need to control the coherence of the column spaces since these matrices are symmetric.}. The second identifiability issue that arises is distinguishing the low rank matrices $L_y^\star$ and $\Theta_{yx}^\star(\Theta_x^\star)^{-1}\Theta_{xy}^\star$ from one another. This task is made difficult when the row/column spaces of these matrices are nearly aligned. Thus, we must ensure that the row/column spaces of $L_y^\star$ and $\Theta_{yx}^\star(\Theta_x^\star)^{-1}\Theta_{xy}^\star$ are sufficiently transverse (i.e. have large angles).\\

 These identifiability issues directly translate to conditions on the population composite factor model. Specifically, $\mu(\text{column-space}(L_y^\star))$ and \\ $\mu(\text{column-space}(\Theta_{yx}^\star(\Theta_x^\star)^{-1}\Theta_{xy}^\star))$ being small translates to $\mu(\text{column-space}(\A^\star))$ and $\mu(\text{column-space}(\B_u^\star))$ being small. Such a condition has another interpretation. It states that the effect of $x$ and $\zeta_u$ must not concentrate on any one variable of $y$; otherwise, this effect can be absorbed by the random variable $\bar{\epsilon}$ in \eqref{eqn:composite}. The second identifiability assumption that the row/column spaces of $L_y^\star$ and $\Theta_{yx}^\star(\Theta_x^\star)^{-1}\Theta_{xy}^\star$ have a large angle translates to the angle between column spaces of $\A^\star$ and $\B_u^\star$ being large. This assumption ensures that the effect of $x$ and $\zeta_u$ on $y$ can be distinguished. \\

Having these identifiability concerns in mind, we give a stylized composite factor model \eqref{eqn:composite} and check that the Fisher Information Assumptions 1,2, and 3 in \eqref{eqn:FirstFisherCond}, \eqref{eqn:FirstFisherCond3}, and \eqref{eqn:SecondFisherCond} are satisfied for appropriate choices of parameters. Specifically, we let $p = 60$, $q = 2$, $k_x = 1$, and $k_u = 1$. We let the random variables $x \in \R^q$, $\zeta_u \in \R^{k_u}$, $\bar{\epsilon} \in \R^p$ be distributed according to $x \sim \mathcal{N}(0,\mathcal{I}_{q{\times}q})$, $\zeta_u \sim \mathcal{N}(0,\mathcal{I}_{k_u{\times}k_u})$, and $\bar{\epsilon} \sim \mathcal{N}(0,\mathcal{I}_{p{\times}p})$. We generate matrices $J \in \R^{p \times k_x}, K \in \R^{q \times k_x}$ with i.i.d Gaussian entries, and let $\A^\star = JK^T$. Similarly, we generate $\B_u^\star \in \R^{p \times k_u}$ with i.i.d Gaussian entries. We scale matrices $\A^\star$ and $\B_{u}^\star$ to have spectral norm equal to $0.2$. With this selection, the smallest angle between the column spaces of $\mathcal{A}^\star$ and $\mathcal{B}^\star_u$ is $87$ degrees. Furthermore, the quantities $\mu(\text{column-space}(\A^\star))$ and $\mu(\text{column-space}(\B_u^\star))$ are $0.072$ and $0.074$ respectively, . Under this stylized setting, we numerically evaluate Assumptions 1, 2, and 3 in \eqref{eqn:FirstFisherCond}, \eqref{eqn:FirstFisherCond3}, and \eqref{eqn:SecondFisherCond} with a Fisher information $\mathbb{I}^\star$ that takes the form:
\begin{eqnarray*}
\mathbb{I}^\star = \begin{pmatrix} \mathcal{I} + \A^\star{\A^\star}^T + \B_u^\star{\B_u^\star}^T  & \A^\star \\ {\A^\star}^T & \mathcal{I} \end{pmatrix} \otimes \begin{pmatrix} \mathcal{I} + \A^\star{\A^\star}^T + \B_u^\star{\B_u^\star}^T  & \A^\star \\ {\A^\star}^T & \mathcal{I} \end{pmatrix}
\end{eqnarray*}
We let $\omega_{y} = 0.03, \omega_{yx} = 0.03$ so that the largest angle between the pair of tangent spaces $T_y', T(L_y^\star)$ and tangent spaces $T_{yx}', T(\Theta_{yx}^\star)$ is less than $1.8$ degrees. Letting $\gamma \in (1,1.4)$, one can numerically check that \\$\inf_{\mathbb{H}' \in U{(\omega_y, \omega_{yx})}}\chi({\mathbb{H}}', \Phi_{\gamma}) > 0.2$,  $\inf_{\mathbb{H}' \in U{(\omega_y, \omega_{yx})}}\Xi({\mathbb{H}}', \Gamma_{\gamma}) > 0.4$ and\\ $\sup_{\mathbb{H}' \in U{(\omega_y, \omega_{yx})}}\varphi({\mathbb{H}}', {{\Gamma}_{\gamma}}) < 0.8$. Thus, for $\omega_y = 0.03$,  $\omega_{yx} = 0.03$, $\alpha = 0.2$, $\beta = 9$, and $\gamma \in (1,1.4)$, the Fisher information condition Assumptions 1, 2, and 3 in \eqref{eqn:FirstFisherCond},  \eqref{eqn:FirstFisherCond3} and \eqref{eqn:SecondFisherCond} are satisfied.


\subsection{{Theorem Statement}}
\label{section:theoremstatement}
We now describe the performance of the regularized maximum-likelihood programs \eqref{eqn:main} under suitable conditions on the quantities introduced in the previous section. Before formally stating our main result, we introduce some notation. Let $\sigma_y$ denote the minimum nonzero singular value of $L_y^\star$ and let $\sigma_{yx}$ denote the minimum nonzero singular value of $\Theta_{yx}^\star$. We state the theorem based on essential aspects of the conditions required for the success of our convex relaxation (i.e. the Fisher information conditions) and omit complicated constants. We specify these constants in Section 4. 

\begin{theorem}
\label{theorem:main}
Suppose that there exists $\alpha > 0$, $\beta \geq 2$, $\omega_y \in (0,1)$, $\omega_{yx} \in (0,1)$, and the choice of parameter $\gamma$ so that the population Fisher information $\mathbb{I}^\star$ satisfies Assumptions 1, 2, and 3 in \eqref{eqn:FirstFisherCond}, \eqref{eqn:FirstFisherCond3} and \eqref{eqn:SecondFisherCond}. Let $m \triangleq \max\{1, \frac{1}{\gamma}\}$, and $\bar{m} \triangleq \max\{1, {\gamma}\}$. Furthermore, suppose that the following conditions hold:
\begin{enumerate}
\item $n_{} \gtrsim \Big[\frac{\beta^2}{\alpha^2} m^{6}\Big] (p+q) $
\item $\lambda_n \sim \Big[\frac{\beta}{\alpha}{m^2}\Big]\sqrt{\frac{p+q}{n}}$
\item $\sigma_{y} \gtrsim \Big[\frac{\beta}{\alpha^5\omega_{y}}{}{m^4}\Big]\lambda_n$ 
\item $\sigma_{yx} \gtrsim \Big[\frac{\beta}{\alpha^5\omega_{yx}}{m^5\bar{m}^2}\Big]\lambda_n$
\end{enumerate}

Then with probability greater than $1-2\exp\Big\{-\tilde{C}_{prob}\frac{\alpha^2}{\beta^2{m}^4} n\lambda_n^2\Big\}$, the optimal solution $(\hat{\Theta},\hat{D}_y,\hat{L}_y)$ of \eqref{eqn:main} with i.i.d. observations $\mathcal{D}^+_{n_{}} = \{y^{(i)}, x^{(i)}\}_{i = 1}^{n}$ of $(y,x)$ satisfies the following properties:
\begin{enumerate}
\item rank($\hat{L}_y$) = rank(${L}_y^\star$), rank($\hat{\Theta}_{yx}$) = rank(${{\Theta}^\star_{yx}}$)\\[.005in]
\item $\|\hat{D}_y - D_y^\star\|_{2} \lesssim \frac{m}{\alpha^2}\lambda_n$, $\|\hat{L}_y - L_y^\star\|_{2} \lesssim \frac{m}{\alpha^2}{}\lambda_n$, $\|\hat{\Theta}_{yx} - \Theta_{yx}^\star\|_{2} \lesssim \frac{m\bar{m}}{\alpha^2}\lambda_n$, $\|\hat{\Theta}_{x} - \Theta_{x}^\star\|_{2} \lesssim \frac{m}{\alpha^2}{}\lambda_n$
\end{enumerate}
\end{theorem}

%
%

We outline the proof of Theorem~\ref{theorem:main} in Section~\ref{section:proofs}. The quantities $\alpha, \beta, \omega_y, \omega_{yx}$ as well as the choices of parameters $\gamma$ play a prominent role in the results of Theorem~\ref{theorem:main}. Indeed larger values of $\alpha, \omega_y, \omega_{yx}$ and smaller values of $\beta$ (leading to a better conditioned Fisher information even for large distortions around the tangent space $T(L_y^\star)$ and $T(\Theta_{yx}^\star)$ lead to less stringent requirements on the sample complexity, on the minimum nonzero singular value of $\sigma_{y}$ of $L_y^\star$, and on the minimum nonzero singular value $\sigma_{yx}$ of $\Theta_{yx}^\star$.

%% file: Experimental.tex
\section{Experimental Results}
In this section, we demonstrate the utility of Algorithm 1 for interpreting latent variables in factor models both with synthetic and real financial asset data. 
\label{section:experiments}
\subsection{{Synthetic Simulations}}
\label{section:simulation}
We give experimental evidence for the utility of Algorithm 1 on synthetic examples. Specifically, we generate a composite factor model \eqref{eqn:composite} $y = \A^\star{x} + \B^\star_u\zeta_u + \bar{\epsilon}$ as follows: we fix $p = 40$ and $q = 10$. We let the random variables $x \in \R^q$, $\zeta_u \in \R^{k_u}$, $\bar{\epsilon} \in \R^p$ be distributed according to $x \sim \mathcal{N}(0,\mathcal{I}_{q{\times}q})$, $\zeta_u \sim \mathcal{N}(0,\mathcal{I}_{k_u{\times}k_u})$, and $\bar{\epsilon} \sim \mathcal{N}(0,\mathcal{I}_{p\times{p}})$. We generate matrices $J \in \R^{p \times k_x}, K \in \R^{q \times k_x}$ with i.i.d Gaussian entries, and let $\A^\star = JK^T$. Similarly, we generate $\B_u^\star \in \R^{p \times k_u}$ with i.i.d Gaussian entries. This approach generates a factor model \eqref{eqn:factormodel} with $k = k_x + k_u$. The composite factor model translates to a joint precision matrix $\Theta^\star$, with the submatrix $\Theta_y^\star = D_y^\star - L_y^\star$ where $D_y^\star$ is diagonal, $\rk(L_y^\star) = k_u$, and $\rk(\Theta_{yx}^\star) = k_x$. We scale matrices $\A^\star$ and $\B_u^\star$ to have spectral norm equal to $\tau$. The value $\tau$ is chosen to be as large as possible without the condition number of $\Theta^\star$ exceeding $10$ (this is imposed for the purposes of numerical conditioning). We obtain four models with $(k_x,k_u) = (1,1), (k_x,k_u) = (2,2)$, and $(k_x,k_u) = (4,4)$, and $(k_x,k_u) = (6,6)$. \\

For the purposes of this experiment, we assume that the input to Algorithm 1 is the oracle factor model specified by the parameters $(D_y^\star, L_y^\star + \Theta_{yx}^\star(\Theta_{x})^{-1}\Theta_{xy}^\star$), and demonstrate the success of steps 2-3 of Algorithm 1. In particular, for each model, we generate $n$ samples of responses $y$ and covariates $x$, and use these observations as input to the convex program \eqref{eqn:main}. The regularization parameters $\lambda_n, \gamma$ are chosen so that the estimates $(\hat{\Theta}, \hat{L}_y, \hat{D}_y)$ satisfy $(i)~\mathrm{rank}(L_y^\star + \Theta_{yx}^\star(\Theta_{x}^\star)^{-1}\Theta_{xy}^\star) = \mathrm{rank}(\hat{L}_y + \hat{\Theta}_{yx}\hat{\Theta}_x^{-1}\hat{\Theta}_{xy})$,\\ $(ii)~\colspace(\hat{\Theta}_{yx}) \cap \colspace(\hat{L}_{y}) = \{0\}$, and the deviation from the underlying factor model \\$\max\{\|{D}_y^\star - \hat{D}_y \|_2 /\|{{D}}_y^\star\|_2, \|{L}_y^\star - [\hat{L}_y+\hat{\Theta}_{yx}\hat{\Theta}_x^{-1}\hat{\Theta}_{xy}]\|_2 / \|{L}_y^\star\|_2\}$ is minimized. Figure 1(a) shows the magnitude of the deviation for different values of $n$. Furthermore, for each fixed $n$, we use the choice of regularization parameters $(\lambda_n, \gamma)$ to compute the probability of obtaining structurally correct estimates of the composite model (i.e. $\rk(\hat{L}_y)  = \rk(L_y^\star)$ and $\rk(\Theta_{yx}^\star) = \rk(\hat{\Theta}_y)$). These probabilities are evaluated over $10$ experiments and are shown in Figure 1(b). These results support Theorem 1 that given (sufficiently many) samples of responses/covariates, the convex program \eqref{eqn:main} provides accurate estimates of the composite factor model \eqref{eqn:composite}. 
%

\FloatBarrier
\begin{figure}[!http]
\centering
\subfigure[composite factor model error]{
\includegraphics[width=5cm, height = 5cm]{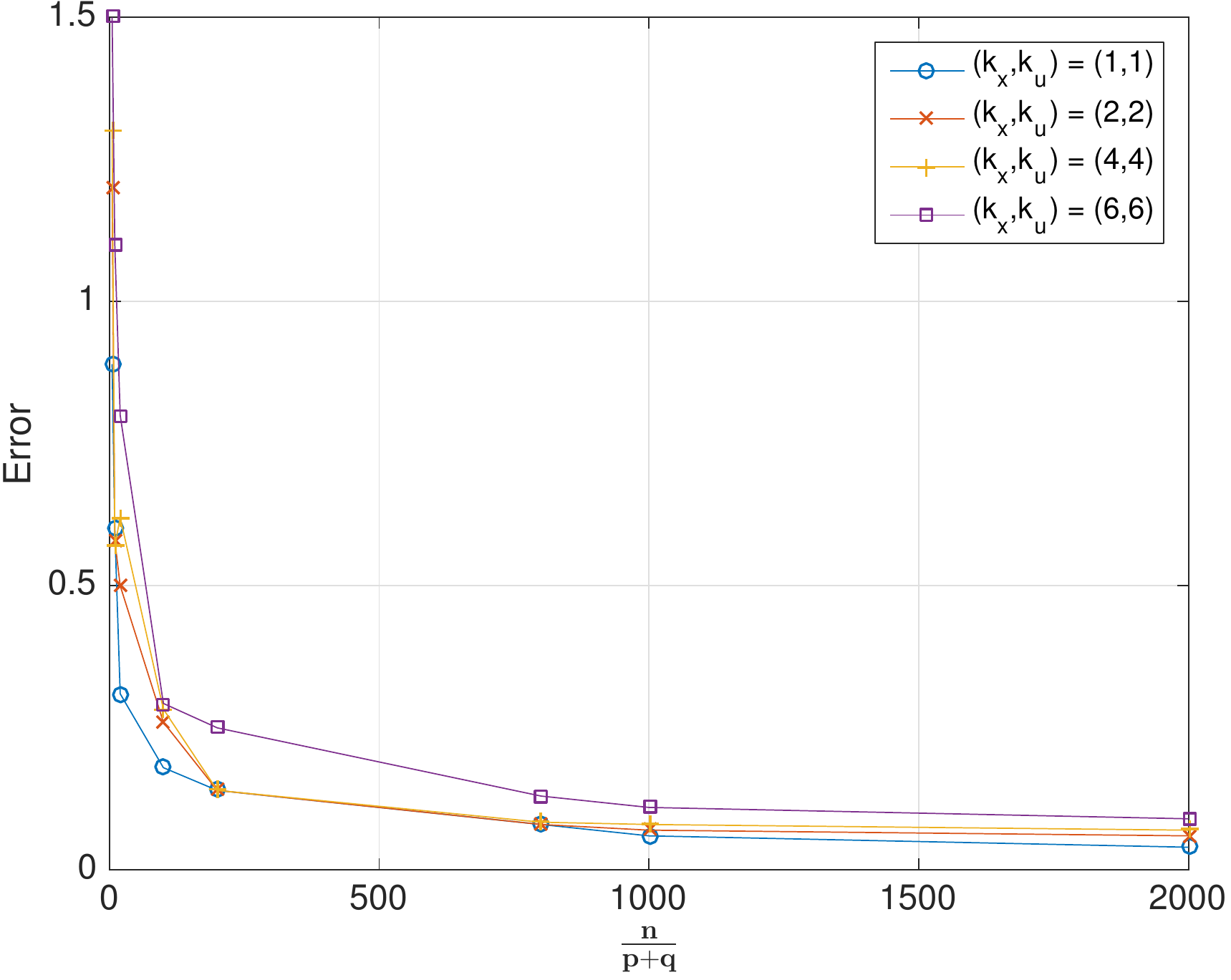}
}\subfigure[composite factor model structural recovery]{
\includegraphics[width=5cm, height = 5cm]{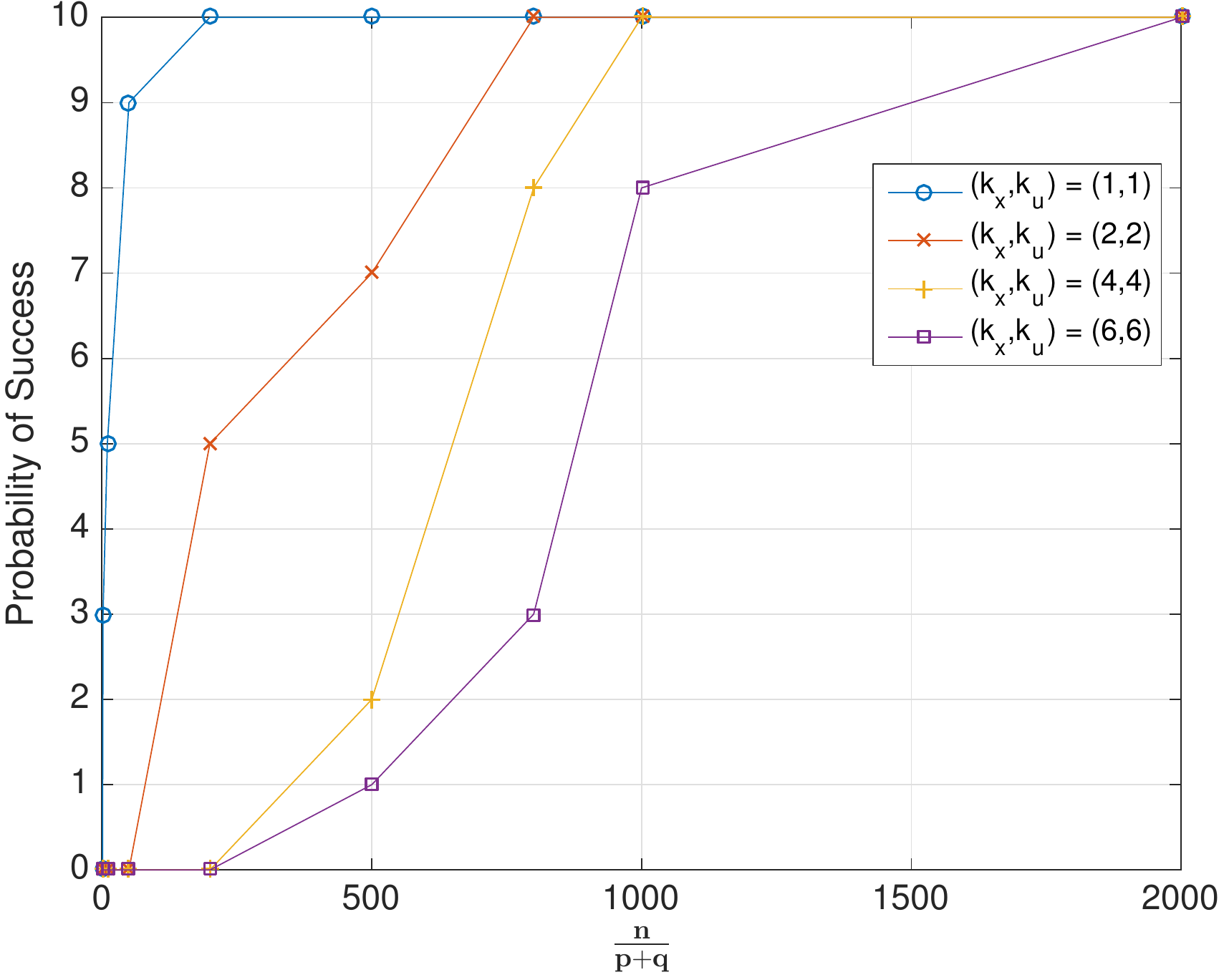}
}
{\caption{Synthetic data: plot shows the error (defined in the main text) and probability of correct structure recovery in composite factor models. The four models studied are $(i)~(k_x, k_u) = (1,1) $, $(ii)~(k_x, k_u) = (2,2)$, and $(iii)~(k_x, k_u) = (4,4)$, and $(iv)~(k_x, k_u) = (6,6)$. For each plotted point in (b), the probability of structurally correct estimation is obtained over $10$ trials.}} 
\end{figure}
\FloatBarrier

\subsection{{Experimental Results on Financial Asset Data}}
We consider as our responses $y$ the monthly stock returns of $p = 45$ companies from the Standard and Poor index over the period March $1982$ to March $2016$, which leads to a total of $n_{} = 408$ observations. We then obtain measurements of $13$ covariates that can plausibly influence the values of stock prices: consumer price index, producer price index, EUR to USD exchange rate, federal debt (normalized by GDP), federal reserve rate, GDP growth rate, government spending (normalized by GDP), home ownership rate, industrial production index, inflation rate, mortgage rate, oil import, and saving rate. Of these $13$ covariates, the covariates federal debt, government spending, GDP growth rate, and home ownership rate are only available at a quarterly scale. Monthly observations are available for the remaining covariates. Evidently, many more observations of $y$ are available than of $(y,x)$ jointly. As described in Section~\ref{section:algorithm}, this scenario motivates us to first learn a factor model using the monthly observations of $y$. We then associate semantics to the latent variables of this factor model by fitting a composite factor model to the more limited joint observations of $(y,x)$.

As a factor model is not available in advance, we begin with learning a factor model~\eqref{eqn:factormodel} using observations of $y$. In particular, we fit observations $\mathcal{D}_n = \{y^{(i)}\}_{i = 1}^n$ to the factor model \eqref{eqn:factormodel} using the following convex relaxation:
\begin{eqnarray}
(\hat{\tilde{D}}_y, \hat{\tilde{L}}_y) = \arg\min_{\substack{\tilde{D}_y,\tilde{L}_y \in \Sp^p \\ \tilde{D}_y - \tilde{L}_y \succ 0}} & -\ell(\tilde{D}_y-\tilde{L}_y; \mathcal{D}_{n}) + \tilde{\lambda}_n \mathrm{trace}(\tilde{L}_y) \nonumber \\ \mathrm{s.t.} & \tilde{L}_y \succeq 0, \tilde{D}_y ~\mathrm{is~diagonal}. &
\label{eqn:main2}
\end{eqnarray}
We note that the convex program \eqref{eqn:main2} is a specialization of the convex program \eqref{eqn:main} for learning a composite factor model. The parameter $\tilde{\lambda}_n$ in \eqref{eqn:main2} provides a tradeoff between fidelity of the model to the observations and the complexity of the model (i.e., the number of latent variables). In contrast to minimum-trace factor analysis -- in which the objective is to decompose a covariance matrix as the sum of a diagonal matrix and a low-rank matrix \cite{Ledermann,Shapiro1,S2,S3}-- the regularized maximum-likelihood convex program \eqref{eqn:main2} fits factor models by decomposing a precision matrix as the difference between a diagonal matrix and a low-rank matrix. Although the focus of this paper is not about learning a factor model accurately, for the sake of completeness, we show in Section 5.6 of the appendix that under suitable conditions on the population model, the convex relaxation \eqref{eqn:main2} provides an accurate estimate of the underlying factor model. \\

For the purpose of learning a factor model, we set aside a random subset of $n_{\train} = 308$ of the total $n_{} = 408$ observations as a training set and the remaining subset of $n_{\test} = 100$ as the test set. We let $\mathcal{D}_\train = \{y^{(i)}\}_{i  = 1}^{n_{\train}}$ and $\mathcal{D}_\test = \{y^{(i)}\}_{i  = 1}^{n_{\test}}$ be the corresponding training and testing data sets respectively. We use the observations $\mathcal{D}_\train$ as input to the convex program~\eqref{eqn:main2} where the regularization parameter $\tilde{\lambda}_n$ is chosen via cross-validation. Concretely, for a particular choice of $\tilde{\lambda}_n$, we supply $\mathcal{D}_\train$ as input to the convex program \eqref{eqn:main2}, and solve \eqref{eqn:main2} to obtain a factor model specified by $(\hat{\tilde{D}}_y , \hat{\tilde{L}}_y)$. We then compute the average log-likelihood over the testing set $\mathcal{D}_\test$ using the distribution specified by the precision matrix $\hat{\tilde{D}}_y - \hat{\tilde{L}}_y$. We perform this procedure as we vary $\tilde{\lambda}_n$ from $0.04$ to $4$ in increments of $0.004$. Figure~\ref{fig:testp} shows a plot of $\text{rank}(\hat{\tilde{L}}_y))$ (i.e. number of latent factors) vs. average log-likelihood performance on the testing set. Notice that fixing the number of latent factors does not lead to a unique factor model as varying the regularization parameter $\tilde{\lambda}_n$ may lead to a change in the estimated model, but no change in its structure (i.e. $\text{rank}(\hat{\tilde{L}}_y)$ remains the same). As larger values of average log-likelihood are indicative of a better fit to test samples, these results suggest that $10$ latent factors influence stock prices. We thus focus on associating semantics to the factor model with the largest average log-likelihood performance that consists of $10$ latent factors.

\FloatBarrier
\begin{figure}[!http]
\centering
\includegraphics[width=8cm, height = 6cm]{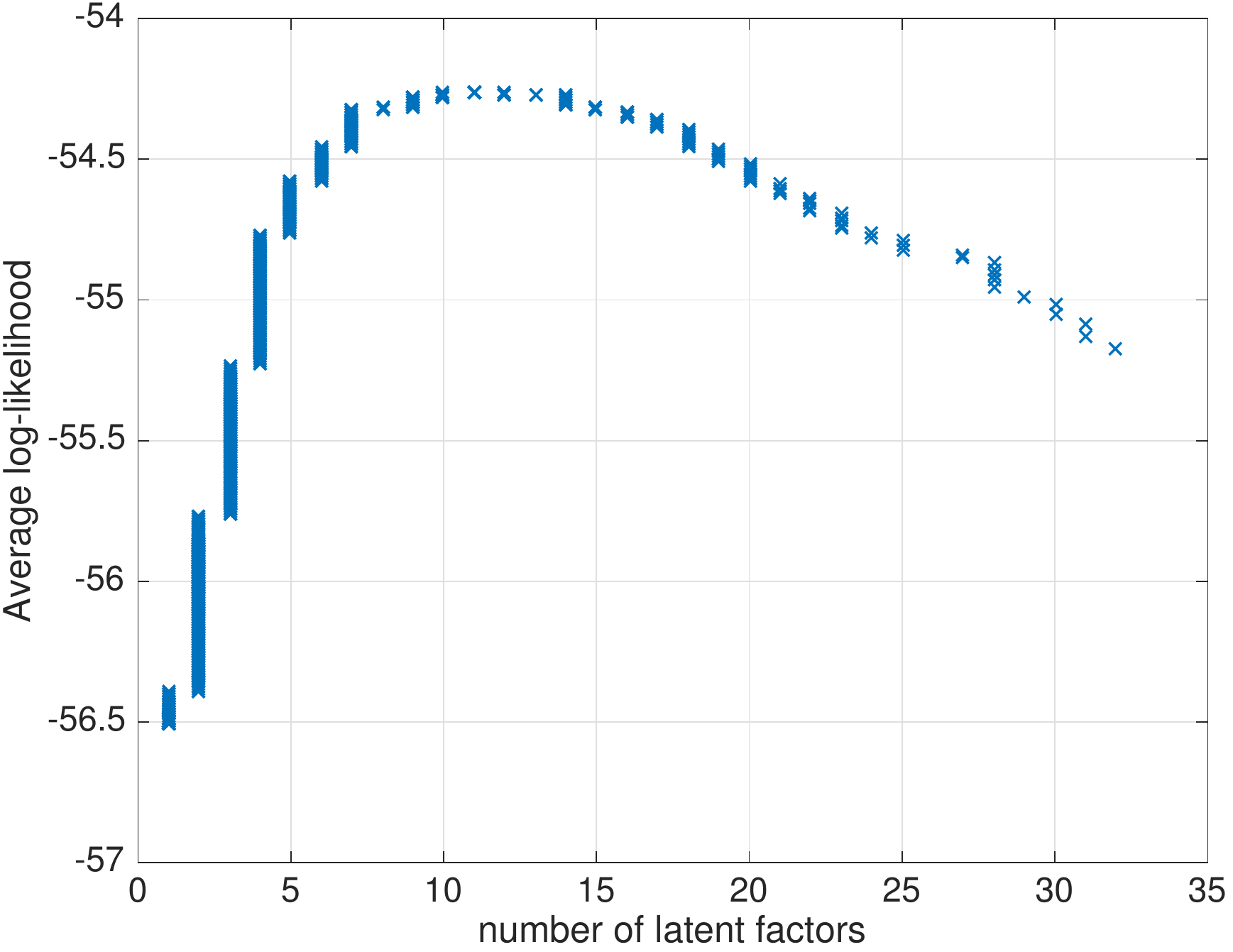}
{\caption{Number of latent factors vs. average log-likelihood over testing set. These results are obtained by sweeping over parameters $\tilde{\lambda}_n \in [0.04, 4]$ in increments of $0.004$ and solving the convex program~\eqref{eqn:main2}}}. 
\label{fig:testp}
\end{figure}
\FloatBarrier

We now proceed with the steps 2-3 of Algorithm 1. To obtain a consistent set of joint observations $(y,x)$ to employ as input to the convex program \eqref{eqn:main}, we apply a 3-month averaging for each variable that is available at a monthly scale (i.e. the responses $y$ and the covariates $x$ with the exception of the four specified earlier) to obtain quarterly measurements. This leads to $n_{} = 137$ quarterly measurements. We denote the quarterly responses and covariates by $\tilde{y}$ and $\tilde{x}$, respectively. We let $\mathcal{D}_n^{+} = \{(\tilde{y}^{(i)},\tilde{x}^{(i)})\}_{i  = 1}^{n}$ be the set of joint quarterly observations of response $\tilde{y}$ and covariates $\tilde{x}$. Using observations $\mathcal{D}_n^+$ as input to the convex program~\eqref{eqn:main} , we perform an exhaustive sweep over parameter space $(\lambda_n, \gamma)$ to learn composite models with estimates $(\hat{\Theta}, \hat{D}_y, \hat{L}_y)$ such that $\text{rank}(\hat{\Theta}) = 0,1,2,\dots 10$, and $\text{rank}(\hat{L}_y) = 0,1,2,\dots 10$. As we are interested comparing these composite models to the factor model with $10$ latent variables, we finely grid the parameter space $(\lambda_n, \gamma)$ so that there are a large number of models for which $\text{rank}(\hat{\Theta}) + \text{rank}(\hat{L}_y)$ is equal to $10$. Among these models, we restrict to those that satisfy the conditions of step 3 of Algorithm 1. Table~\ref{table:nummodels} shows the number of models that satisfy these conditions for $\text{rank}(\hat{\Theta}_{yx}) = 1,\dots,5$. For each $d = 1,\dots,5$, we then identify the composite factor model which minimizes the quantity $\max\{\|\hat{\tilde{D}}_y - \hat{D}_y \|_2 /\|\hat{\tilde{D}}_y\|_2, \|\hat{\tilde{L}}_y - \hat{L}_y-\hat{\Theta}_{yx}\hat{\Theta}_x^{-1}\hat{\Theta}_{xy}]\|_2 / \|\hat{\tilde{L}}_y\|_2\}$. Table 2 shows the values of this quantity for $\text{rank}(\hat{\Theta}_{yx}) = 1,\dots,5$ with respect to the factor model with $10$ latent variables.
\FloatBarrier
\begin{table}[ht]
\centering 
\begin{tabular}{c c} 
\hline 
$(\text{rank}(\hat{\Theta}_{yx}), \text{rank}(\hat{L}_{y}))$ & $\#$ models satisfying conditions of step 2. \\
\hline 
(1,9) & 167\\ 
(2,8) & 196 \\
(3,7) & 218\\
(4,6) & 110 \\
(5,5) & 98 \\
\hline 
\end{tabular}
{\caption{Number of composite factor models with $\text{rank}(\hat{\Theta}_{yx}) = 1,\dots,5$ that satisfy the requirements of step 2 in Algorithm 1(for the factor model with $10$ latent variables).}}
\label{table:nummodels} 
\end{table}
\FloatBarrier
\FloatBarrier
\begin{table}[ht]
\centering 
\begin{tabular}{c c} 
\hline 
$(\text{rank}(\hat{\Theta}_{yx}), \text{rank}(\hat{L}_{y}))$ & $\max\{\|\hat{\tilde{D}}_y - \hat{D}_y \|_2 /\|\hat{\tilde{D}}_y\|_2, \|\hat{\tilde{L}}_y - \hat{L}_y-\hat{\Theta}_{yx}\hat{\Theta}_x^{-1}\hat{\Theta}_{xy}]\|_2 / \|\hat{\tilde{L}}_y\|_2\}$\\[0.5ex]
\hline 
(1,9)& 0.39 \\ 
(2,8) & 0.40 \\
(3,7) & 0.47 \\
(4,6) & 0.51 \\
(5,5) & 0.55 \\ [1ex] 
\hline 
\end{tabular}
\label{table:deviation12} 
{\caption{Deviation of the candidate composite factor model from the factor model consisting of $10$ latent variables.} } 
\end{table}
\FloatBarrier

Examining Table 2, we note that there is large increase in deviation as $\text{rank}(\hat{\Theta}_{yx})$ is increased above $2$. Thus, we consider the composite factor model with $\text{rank}(\hat{\Theta}_{yx}) = 2$ to be an acceptable approximation of the underlying factor model. As a final step of the algorithm, we investigate the properties of the two-dimensional row-space of $\hat{\Theta}_{yx}$ to shed some light on those covariates that appear to play a significant role in capturing some of the latent phenomena in the $10$-factor model.  In particular, for the composite factor model with $(\text{rank}(\hat{\Theta}_{yx}), \text{rank}(\hat{L}_{y})) = (2,8)$ (second row in Table 2), we let $V \in \R^{13\times 2}$ denote a matrix with orthogonal, unit-norm columns such that the columns of $V$ form a basis for the row space of $\hat{\Theta}_{yx}$ (such a matrix may be computed, for example, via the singular value decomposition).  Thus, the projection of $x$ onto the row-space of $\hat{\Theta}_{yx}$ -- given by $V^Tx$ -- represents the interpretable component of the latent variables. We then consider the Euclidean-squared-norm  of the $i$-th row of $V$, as this specifies the relative strength of the $i$-th covariate. As shown in Table~\ref{table:covariaterelevance}, all covariates have some contribution (as we allow general linear combinations of the covariates $x$ in the composite factor model~\eqref{eqn:composite}). However, the covariates exchange rate, government expenditures, and GDP growth rate seem to be the most relevant, and the covariates mortgage rate and oil import seem to be the least relevant.
\FloatBarrier
\begin{table}[ht]
\centering 
\begin{tabular}{c c} 
\hline\hline 
covariate & strength\\ [0.5ex] 
\hline 
Exchange rate & 0.18 \\
Government expenditures & 0.14 \\
GDP growth rate & 0.11 \\
Home ownership rate & 0.09\\
Industrial production rate & 0.08 \\
PPI & 0.08 \\
CPI & 0.07 \\
Federal debt & 0.06 \\
Saving rate & 0.04 \\
Inflation rate & 0.04\\
Federal reserve rate & 0.03\\
Oil import & 0.03 \\
Mortgage rate & 0.01 \\
\hline 
\end{tabular}
{\caption{Strength of each covariate in the composite factor model with $2$-dimensional projection of covariates and $8$ latent variables} 
\label{table:covariaterelevance}} 
\end{table}
\FloatBarrier

%% file: ProofStrategy.tex
\section{Proof Strategy of Theorem 1}
\label{section:proofs}
We first begin by specifying the constants in Theorem 1. Let $\psi \triangleq \|{\Theta^\star}^{-1}\|_2$, $\tilde{C} = 352\psi^3$, $\tilde{C}_0 = \max\{\frac{1}{192\psi}, 2\psi, \frac{1}{24\psi^2\max\{\frac{2}{\psi^2}+8,\frac{1}{\psi}\}}, \frac{\psi}{8}\}$, $\tilde{C}_{samp} = \tilde{C}\tilde{C}_0$, $\tilde{C}_1 =\frac{1}{6}(186\psi^2+56\psi^4)$, $\tilde{C}_{\sigma} = 6\psi^4(56\psi^4+186\psi^2)^2$, and $\tilde{C}_{prob} = \frac{1}{247808\psi^6}$. The precise conditions on the number of observations, the regularization parameter $\lambda_n$, minimum nonzero singular value of $L_y^\star$ and minimum nonzero singular value of $\Theta_{yx}^\star$ for Theorem 1 are given by:
\begin{enumerate}
\item $n \geq \tilde{C}_{samp}^2\Big[\frac{\beta^4}{\alpha^2}m^6(p+q)\Big]$
\item $\lambda_n \in \Bigg[{\tilde{C}_{}}\Big\{\frac{\beta}{\alpha}{m^2}\sqrt{\frac{p+q}{n}}\Big\}, \frac{1}{\beta{m}{\tilde{C}_{0}}}\Bigg]$
\item $\sigma_y \geq \tilde{C}_{\sigma}\Big[\frac{\beta}{\alpha^5\omega_y} m^4\lambda_n\Big]$
\item  $\sigma_{yx} \geq \tilde{C}_{\sigma_{}}\Big[\frac{\beta}{\alpha^5\omega_{yx}} m^5\bar{m}^2\lambda_n\Big]$
\end{enumerate}
Moreover, under these conditions, with probability greater than \\ $1-2\exp\Big(-\tilde{C}_{prob}\frac{\alpha^2}{m^4\beta^2}n\lambda_n^2\Big)$, the optimal solution of the convex program \eqref{eqn:main} with estimates $(\hat{\Theta}, \hat{L}_y, \hat{D}_y)$ satisfies the following properties:
\begin{enumerate}
\item rank($\hat{L}_y$) = rank(${L}_y^\star$), rank($\hat{\Theta}_{yx}$) = rank(${{\Theta}^\star_{yx}}$)\\[.005in]
\item $\|\hat{D}_y - D_y^\star\|_{2} \leq \tilde{C}_1\frac{m}{\alpha^2}\lambda_n$, $\|\hat{L}_y - L_y^\star\|_{2} \leq \tilde{C}_1 \frac{m}{\alpha^2}\lambda_n$, $\|\hat{\Theta}_{yx} - \Theta_{yx}^\star\|_{2}\leq \tilde{C}_1\frac{m\bar{m}}{\alpha^2}\lambda_n$, $\|\hat{\Theta}_{x} - \Theta_{x}^\star\|_{2} \leq \tilde{C}_1\frac{m}{\alpha^2}\lambda_n$
\end{enumerate}

Now under assumptions of Theorem 1, we construct appropriate primal feasible variables $(\hat{\Theta},\hat{D}_y, \hat{L}_y)$ that satisfy the conclusions of the theorem - i.e., $\hat{\Theta}_{yx}$, $\hat{L}_y$ are low-rank (with the same ranks as the underlying population quantities $\Theta_{yx}^\star$ and $L_y^\star$) - and for which there exists a corresponding dual variable certifying optimality. This proof technique is sometimes also referred to as a primal-dual witness or certificate approach \cite{Wai2009}. The high-level proof strategy is similar in spirit to the proofs of consistency results for sparse graphical model recovery \cite{Ravikumar} and latent variable graphical model recovery \cite{Chand2012}, although our convex program and the conditions required for its success are different from these previous results. Consider the following convex program

\begin{eqnarray}
(\hat{\Theta}, \hat{D}_y, \hat{L}_y) = \arg\min_{\substack{\Theta \in \Sp^{p+q}, ~\Theta \succ 0 \\ D_y,L_y \in \Sp^p}} & -\ell(\Theta; \mathcal{D}_n^+) + \lambda_n [\gamma\|\Theta_{yx}\|_{\star} + \|L_y\|_{\star}] \nonumber \\ \mathrm{s.t.} & \Theta_{y} = D_y - L_y, D_y ~\mathrm{is~diagonal}
\label{eqn:ConvexRelaxed_N}
\end{eqnarray}

Comparing \eqref{eqn:ConvexRelaxed_N} with the convex program \eqref{eqn:main}, the difference is that we no longer constrain ${L}_y$ to be a positive semidefinite matrix. In particular, if ${L}_y \succeq 0$, then the nuclear norm of the matrix ${L}_y$ in the objective function of \eqref{eqn:ConvexRelaxed_N} reduces to the trace of $L_y$. We show in the appendix that with high probability, the matrix $\hat{L}_y$ is positive semidefinite. Standard convex analysis states that $(\hat{\Theta}, \hat{D}_y, \hat{L}_y)$ is the solution of the convex program \eqref{eqn:ConvexRelaxed_N} if there exists a dual variable $\Lambda \in \Sp^p$ with the following optimality conditions being satisfied:
\begin{eqnarray*}
[\Sigma_n - {\hat{\Theta}}^{-1}]_y + \Lambda = 0&;& \hspace{.1in} [\Sigma_n - {\hat{\Theta}}^{-1}]_y \in \lambda_n\partial\|\hat{L}_y\|_\star\\[.01in]
[\Sigma_n - {\hat{\Theta}}^{-1}]_{yx}  \in -\lambda_n\gamma\partial\|\hat{\Theta}_{yx}\|_{\star}&;&  \hspace{.1in} [\Sigma_n - {\hat{\Theta}}^{-1}]_{x}  = 0\\ \hat{\Theta}_y = \hat{D}_y - \hat{L}_y; \hspace{.1in} \hat{D}_y \text{ is diagonal}&;& \hspace{.1in} \Lambda_{i,i} = 0 ~ \text{ for } i = 1,2,\dots p
\end{eqnarray*}
Recall that elements of the subdifferential with respect to nuclear norm at a matrix $M$ have the key property that they decompose with respect to the tangent space $T(M)$. Specifically, the subdifferential with respect to the nuclear norm at a matrix $M$ with (reduced) SVD given by $M = UQV^T$ is as follows:
\begin{eqnarray*}
N \in \partial\|M\|_{\star} \Leftrightarrow \mathcal{P}_{T(M)} (N) = UV^T ~, ~ \|\mathcal{P}_{T(M)^\perp} (N)\|_2 \leq 1,
\end{eqnarray*}
where $\mathcal{P}$ denote a projection operator. Let us denote the subspace $\mathcal{W} \in \Sp^p$ as the set of diagonal matrices with nonnegative entries. Let SVD of $\hat{L}_y$ and $\hat{\Theta}_{yx}$ be given by $\hat{L}_y = \bar{U}\bar{Q}\bar{V}'$ and $\hat{\Theta}_{yx} = \breve{U}\breve{Q}{\breve{V}}'$ respectively, and $Z \triangleq  (0, \hspace{.1in} \lambda_n\bar{U}\bar{V}', \hspace{.1in}  -\lambda_n\gamma_{}{\breve{U}}{\breve{V}}',  \hspace{.1in} 0)$. Setting $\Lambda  = [\Sigma_n - \hat{\Theta}^{-1}]_{Y, \text{off diagonal}}$, and letting $\mathbb{H} = \mathcal{W} \times T(\hat{L}_y) \times T(\hat{\Theta}_{yx}) \times \Sp^q$, the optimality conditions of \eqref{eqn:ConvexRelaxed_N} can be reduced to:
\begin{center}
\begin{enumerate}
\item $\mathcal{P}_{\mathbb{H}}\mathcal{F}^{\dagger}(\Sigma_n - \hat{\Theta}^{-1}) = Z$
\item $\|\mathcal{P}_{T(\hat{L}_y)^\perp} (\Sigma_n - \hat{\Theta}^{-1})_y\|_2 < \lambda_n$;  $\|\mathcal{P}_{T(\hat{\Theta}_{yx})^\perp} (\Sigma_n - \hat{\Theta}^{-1})_{yx}\|_2 < \lambda_n\gamma$
\end{enumerate}
\end{center}

Our analysis proceeds by constructing variables $(\hat{\Theta}, \hat{D}_y, \hat{L}_y)$ that satisfy the optimality conditions specified above. Consider the optimization program \eqref{eqn:ConvexRelaxed_N} with additional (non-convex) constraints that $L_y$ and $\Theta_{yx}$ belong to algebraic variety of low rank matrices specified by $L_y^\star$ and $\Theta_{yx}^\star$. While this new program is non-convex, it has a very interesting property that at the global optimal solution (and indeed at any locally optimal solution) $\hat{L}_y$ and $\hat{\Theta}_{yx}$ are smooth points of their respective algebraic varieties. This observation suggests that the Lagrange multipliers corresponding to the additional variety constraints belongs to $T(\hat{L}_y)^\perp$ and $T(\hat{\Theta}_{yx})^\perp$ respectively. We show under suitable conditions that $(\hat{\Theta}, \hat{D}_y, \hat{L}_y)$ also satisfy the second optimality condition of \eqref{eqn:ConvexRelaxed_N} corresponding to the tangent spaces $T(\hat{L}_y)^\perp$ and $T(\hat{\Theta}_{yx})^\perp$. Thus $(\hat{\Theta}, \hat{D}_y, \hat{L}_y)$ is a unique solution of \eqref{eqn:main} and as constructed, is algebraically consistent (i.e. $\text{rank}(\hat{L}_y) = \text{rank}(L_y^\star)$ and $\text{rank}(\hat{\Theta}_{yx}) = \text{rank}(\Theta_{yx}^\star)$)

\subsection{Results Proved in the Appendix}
To ensure that the estimate $\hat{\Theta}$ is close to the population quantity $\Theta^\star$, the quantity $E = \hat{\Theta} - \Theta^\star$ must be small. Since the optimality conditions of \eqref{eqn:ConvexRelaxed_N} are stated in terms of $\hat{\Theta}^{-1}$, we bound the deviation between $\hat{\Theta}^{-1}$ and ${\Theta^\star}^{-1}$. Specifically, the Taylor series expansion of $\hat{\Theta}^{-1}$ around $\Theta^\star$ is given by:
\begin{eqnarray*}
\hat{\Theta}^{-1} = (\Theta^\star+ E)^{-1} = {\Theta^{\star}}^{-1} + {\Theta^{\star}}^{-1}E{\Theta^{\star}}^{-1} + R_{\Sigma^\star}(E)
\end{eqnarray*}
where, $R_{\Sigma^\star}(E) = \Sigma^\star\Big[\sum_{k = 2}^{\infty}(-E\Theta^\star)^k\Big]$. Recalling that $\mathbb{I}^\star = {\Theta^\star}^{-1} \otimes {\Theta^\star}^{-1}$, we note that $\hat{\Theta}^{-1} - {\Theta^\star}^{-1} = \mathbb{I}^\star({E}) + {R}_{\Sigma^\star}(E)$.  In Section~\ref{section:Fishercond},  we imposed assumptions 1,2, and 3 in \eqref{eqn:FirstFisherCond}, \eqref{eqn:FirstFisherCond3}, and \eqref{eqn:SecondFisherCond} on $\mathbb{I}^\star$. These assumptions allow us to control $\mathbb{I}^\star(E)$ when $E$ is restricted to certain directions. We bound the remainder term ${R}_{\Sigma^\star}(E)$ in Proposition~\ref{prop:Remainder} where $E$ is restricted to live in a certain space. Specifically, consider the following constrained optimization program:
\begin{eqnarray}
(\tilde{\Theta}, \tilde{D}_y, \tilde{L}_y) = \argmin_{\substack{\Theta \in \Sp^{q+p}, ~\Theta \succ 0 \\ D_y,{L}_y \in \Sp^p}} & -\ell(\Theta; \mathcal{D}_+^n) + \lambda_n [\|{L_y}\|_{\star} + \gamma \|\Theta_{yx}\|_\star] \nonumber \\ \mathrm{s.t.} ~~\hspace{-0.05in} & \Theta_y = D_y - {L}_y, ~(D_y, {L}_y, \Theta_{yx}, \Theta_{x}) \in \mathbb{H}' \label{eqn:ConvexRelaxed231_N}
\end{eqnarray}
Here $\mathbb{H}' = \mathcal{W} \times T_y' \times T_{yx}' \times \Sp^{q}$, where $T_y'$  is a subspace in $\Sp^{p}$, and $T_{yx}'$ is a subspace in $\R^{p \times q}$. Let $\Delta = (\tilde{D}_y - D_y^\star, \tilde{L}_y - L_y^\star, \tilde{\Theta}_{yx} - \Theta_{yx}^\star, \tilde{\Theta}_x - \Theta_x^\star)$ denote the error in the estimated variables. Furthermore, let $\Delta_1 = \tilde{D}_y - D_y^\star$, $\Delta_2 = \tilde{L}_y - L_y^\star$ and so forth. In the following proposition, we bound the remainder term $R_{\Sigma^\star}(\mathcal{F}(\Delta))$ defined earlier. 

\begin{proposition}
\label{prop:Remainder}
Let $\psi \triangleq \|{\Theta^\star}^{-1}\|_2$ and $C' = (3+\gamma)\psi$. If $\Phi_{\gamma}[\Delta] \leq \frac{1}{2C'}$, then $\Phi_{\gamma}[\mathcal{F}^{\dagger}R_{\Sigma^\star}(\mathcal{F}(\Delta))] \leq 2m{\psi}C'^2 \Phi_{\gamma}[\Delta]^2$.
\end{proposition}
Notice the bound on $R_{\Sigma^\star}(\mathcal{F}(\Delta))$ is dependent on the error term $\Phi_{\gamma}[\Delta]$. In the following proposition, we bound this error so that we can control the remainder term. Suppose that for $\alpha > 0$, $\beta \geq 2$, $\omega_{y} \in (0,1)$, and $\omega_{yx} \in (0,1)$, the Fisher information conditions \eqref{eqn:FirstFisherCond}, \eqref{eqn:FirstFisherCond3}, and \eqref{eqn:SecondFisherCond} are satisfied. Suppose we let $T_y'$ and $T_{yx}'$ be tangent spaces to the low-rank matrix varieties and $\rho(T_y', T(L_y^\star)) \leq \omega_y$ and $\rho(T_{yx}', T(\Theta_{yx}^\star)) \leq \omega_{yx}$. Let $E_n = \Sigma^\star - \Sigma_n$ denote the difference between the true joint covariance and the sample covariance and let $C_T = (\mathcal{P}_{{T_y'}^\perp}(L_y^\star), \mathcal{P}_{{T_{yx}'}^\perp}(\Theta_{yx}^\star))$. The proof of the following result uses Brouwer's fixed-point theorem, and is inspired by the proof of a similar result in \cite{Ravikumar,Chand2012}.
\begin{proposition}
\label{prop:Brower}
 Let $\kappa \triangleq \beta(3+\frac{16}{\alpha}\psi^2m)$. Consider the following two quantities:
\begin{eqnarray}
r_1 &\triangleq& \max\Big\{\frac{4}{\alpha}\Big(\Phi_{\gamma}[\mathcal{F}^{\dagger}E_n] + \Phi_{\gamma}[\mathcal{F}^{\dagger}\mathbb{I}^{\star}\mathcal{F}C_{T_{}}] +\lambda_n\Big),\hspace{.1in} \Phi_{\gamma}[C_{T}]\Big\} \\
r_2 &\triangleq&  \frac{4}{\alpha}\Big(\Phi_{\gamma}[\mathcal{F}^{\dagger}E_n] + \Phi_{\gamma}[\mathcal{F}^{\dagger}\mathbb{I}^{\star}\mathcal{F}C_{T_{}}]\Big) 
\label{eqn:rdef}
\end{eqnarray}
Define $r_1^u \triangleq \max\Big\{\frac{4}{\alpha}\Big(\frac{2\lambda_n}{\kappa}+\lambda_n\Big),\hspace{.1in} \frac{\lambda_n}{\kappa}\Big\}$ and $r_2^u \triangleq \frac{8\lambda_n}{\alpha\kappa}$. Suppose that 1) $r_1 \leq r_1^u$, 2) $r_2 \leq r_2^u$, and 3) $r_1^u \leq \min\{\frac{1}{4C'}, \frac{\alpha}{32\max\{1+\frac{\kappa}{2},\frac{\alpha}{8}\}^2m{\psi}C'^2}\}$, \\ then  $\max\{\|\Delta_2\|_{2}, \frac{1}{\gamma}\|\Delta_3\|_2\} \leq 2r_1^u$ and $\max\{\|\Delta_1\|_2, \|\Delta_4\|_2\} \leq r_2^u$. Consequently, $\Phi_{\gamma}(\Delta) \leq 2r_1^u$. 
\end{proposition}

In the following proposition, we prove algebraic correctness of program \eqref{eqn:ConvexRelaxed231_N}. The statement of this proposition requires us to define some constants. Let $C_1' = \frac{2m}{\alpha\kappa} \Big(28(1+\frac{4}{\alpha})\psi^2 + 2\kappa + 13\Big)$, $C_2' = \frac{4}{\alpha}(\frac{1}{2\kappa}+1)$,  $C_{\sigma_{y}}' = C_1'^2\psi^2\max\{2\kappa + 1, \frac{2}{C_2'\psi^2} + 1\}$, $C_{\sigma_{yx}}' = C_1'^2\psi^2\max\{2\kappa + \frac{\kappa}{\gamma}, \frac{2}{C_2'\psi^2} + \frac{\kappa}{\gamma}\}$, and \\$C_{samp}' = \max\{\frac{1}{8m\psi\kappa}, \frac{\alpha}{16C'(\frac{2}{\kappa}+1)}, \frac{\alpha^2}{128(\frac{2}{\kappa}+1)\max\{1+\frac{\kappa}{2},\frac{\alpha}{8}\}^2m\psi^2C'^2}, \frac{1}{4C_1'C'}\}$
\begin{proposition}
\label{proposition:lambdaBound}
Suppose that $\sigma_y \geq  \frac{{m}}{\omega}C_{\sigma_y}'\lambda_n$, $\sigma_{yx} \geq  {m\gamma^2}C_{\sigma_{yx}}'\lambda_n$. Further, suppose that $\lambda_n$ is chosen so that $\lambda_n \leq \frac{1}{C_{samp}'}$. Then, there exists tangent space $T_y' \subset \Sp^{p}$ in the rank-$k_{u}$ variety ($k_u = \text{rank}(L_y^\star)$) and tangent space $T_{yx}' \subset \R^{p \times q}$ in rank $k_{x}$-variety ($k_x = \text{rank}(\Theta_{yx}^\star)$) where $\rho(T_y', T(L_y^\star)) \leq \omega_y$, $\rho(T_{yx}', T(\Theta_{yx}^\star)) \leq \omega_{yx}$ such that the corresponding solution $(\tilde{\Theta}, \tilde{D}_y, \tilde{L}_y)$ of \eqref{eqn:ConvexRelaxed231_N} satisfies the following properties:
\begin{enumerate}
\item$\text{rank}(\tilde{L}_y) = \text{rank}(L_y^\star)$ and $\text{rank}(\tilde{\Theta}_{yx}) = \text{rank}(\Theta_{yx}^\star)$
\item Letting $C_{T} = (0~,~\mathcal{P}_{{T_y'}^\perp}(L_y^\star)~, ~\mathcal{P}_{{T_{yx}'}^\perp}(\Theta_{yx}^\star)~,~ 0)$, we have that $\Phi_{\gamma}[\mathcal{F}^\dagger \mathbb{I}^\star\mathcal{F}(C_{T})] \leq \frac{\lambda_n}{\kappa}$ and $\Phi_{\gamma}[C_T] \leq \frac{4}{\alpha}(1+\frac{2}{\kappa})\lambda_n$
\item $\Phi_{\gamma}[\Delta] \leq 2C_1'\lambda_n$
\item $\tilde{L}_y \succeq 0$
\end{enumerate}
Furthermore, suppose that $\Phi_{\gamma}(\F^{\dagger}E_n) \leq \frac{\lambda_n}{\kappa}$ and $~\Phi_{\gamma}[\mathcal{F}^{\dagger}R_{\Sigma^\star}(\mathcal{F}(\Delta))] \leq \frac{\lambda_n}{\kappa}$. Then the tangent space constraint $(D_y, L_y, \Theta_{yx}, \Theta_x) \in \mathbb{H}'$ in \eqref{eqn:ConvexRelaxed231_N} is inactive, so that  $(\tilde{\Theta}, \tilde{D}_y, \tilde{L}_y)$ is the unique solution of the original convex program \eqref{eqn:main}.
\label{prop:rec}
\end{proposition}

Thus far, the analysis of the convex program so has been deterministic in nature. In the following proposition, we present the probabilistic component of our analysis by showing the rate at which the sample covariance matrix $\Sigma_n$ converges to $\Sigma^\star$ in spectral norm. This result is well-known and is a specialization of a result proven by \cite{Davidson}.
\begin{proposition}
\label{prop:Enbound}
Suppose that the number of observed samples obeys \\$n \geq 64 \kappa^2m^2\psi^2C_{samp}'^2(p+q)$, and the regularization parameter $\lambda_n$ is chosen so that: \\ $\lambda_n \in [8\psi\kappa{m}\sqrt{\frac{p+q}{n}}, \frac{1}{C_{samp}'}]$. Then, with probability greater than $1 - 2\text{exp}\Big\{-\frac{n\lambda_n^2}{128\kappa^2m^2\psi^2}\Big\}$, $\Phi_{\gamma}[\mathcal{F}^{\dagger}E_n] \leq \frac{\lambda_n}{\kappa}$. \label{prop:prob} \end{proposition}

\subsection{Proof of Theorem 1} We first relate the constants $\tilde{C}_{samp}$, $\tilde{C}$, $\tilde{C}_0$, $\tilde{C}_1$, and $\tilde{C}_{\sigma}$ of Theorem 1 to the constants $C_{samp}'$, $C_1'$ $C'_{\sigma_y}$, and $C'_{\sigma_{yx}}$. In particular, using the properties that $\beta \geq 2$ and $\frac{\psi^2}{\alpha} \geq \frac{1}{2}$ and $\bar{m}, m \geq 1$,  one can check that:  $\tilde{C}_{samp}^2 \geq \frac{64\psi^2\kappa^2\alpha^2}{\beta^4m^6}C_{samp}'^2$, $\tilde{C}_0 \geq \frac{1}{\beta{m}}C_{samp}'$, $\tilde{C}_{\sigma_{}} \geq \frac{\alpha^5}{\beta{m}^3}C_{\sigma_y}'$, $\tilde{C}_{\sigma_{}} \geq \frac{\alpha^5}{\beta{m}^4}C_{\sigma_{yx}}'$, and $\tilde{C}_1 \geq \frac{\alpha^2}{m}C'_1$. Furthermore, we have that $\tilde{C} \geq \frac{\alpha}{\beta{m}}8\psi\kappa$. Using these relations, one can also check that the assumptions of Theorem 1 imply that the assumptions of Proposition~\ref{prop:rec} and Proposition 4 are satisfied. Thus we can conclude that the optimal solution $(\tilde{\Theta}, \tilde{D}_y, \tilde{L}_y)$ of \eqref{eqn:ConvexRelaxed231_N} (with a particular choice of tangent spaces $T_y'$ and $T_{yx}'$) satisfy results of Proposition~\ref{prop:rec}. Further, by appealing to Proposition~\ref{prop:prob}, we have that $\Phi_{\gamma}(\mathcal{F}^\dagger{E_n})  \leq \frac{\lambda_n}{\kappa}$. If we show that $~\Phi_{\gamma}[\mathcal{F}^{\dagger}R_{\Sigma^\star}(\Delta)] \leq \frac{\lambda_n}{\kappa}$, then we conclude that the unique optimum $(\hat{\Theta}, \hat{D}_y, \hat{L}_y)$ of the original convex program \eqref{eqn:main} coincide with the optimum $(\tilde{\Theta}, \tilde{D}_y, \hat{L}_y)$ of the convex program \eqref{eqn:ConvexRelaxed231_N}. Thus, we conclude that the estimates of \eqref{eqn:main} have structurally correct structure (i.e. $\text{rank}(\hat{L}_y) = \text{rank}(L_y^\star)$ and $\text{rank}(\hat{\Theta}_{yx}) = \text{rank}(\Theta_{yx}^\star)$) and have their error bounded by $\Phi_{\gamma}(\Delta) \leq 2C_1'\lambda_n$. To show that  $~\Phi_{\gamma}[\mathcal{F}^{\dagger}R_{\Sigma^\star}(\Delta)] \leq \frac{\lambda_n}{\kappa}$, we note that
\begin{eqnarray*}
\frac{4}{\alpha}\Big(\Phi_{\gamma}[\mathcal{F}^{\dagger}E_n] + \Phi_{\gamma}[\mathcal{F}^{\dagger}\mathbb{I}^{\star}\mathcal{F}C_{T}] +\lambda_n \Big) \leq \frac{4}{\alpha} \Big(\frac{\lambda_n}{\kappa} + \frac{\lambda_n}{\kappa} + \lambda_n\Big)
\leq \frac{4\lambda_n}{\alpha}\Big(\frac{2}{\kappa}+1\Big)\\ \leq \min\Big\{\frac{1}{4C'}, \frac{\alpha}{32\max\{1+\frac{\kappa}{2},\frac{\alpha}{8}\}^2m{\psi}C'^2}\Big\}
\end{eqnarray*}
Here, we used the bound on $\Phi_{\gamma}[\mathcal{F}^{\dagger}\mathbb{I}^{\star}\mathcal{F}C_{T}]$ provided by Proposition~\ref{prop:rec} and the bound on $\lambda_n$. Furthermore, appealing to Proposition~\ref{prop:rec} once again, we have $\Phi_{\gamma}[C_{T}] \leq 
\frac{4}{\alpha}(1+\frac{2}{\kappa})\lambda_n \leq \min\{\frac{1}{4C'}, \frac{\alpha}{16m{\psi}C'^2}\}$. Thus Proposition~\ref{prop:Brower} provides us with the bound $\Phi_{\gamma}[\Delta] \leq 2C_1'\lambda_n \leq \frac{1}{2C'}$.  We subsequently apply the results of Proposition~\ref{prop:Remainder} to obtain:
\begin{eqnarray*}
\Phi_{\gamma}[\mathcal{F}^{\dagger}R_{\Sigma^\star}(\mathcal{F}(\Delta))] \leq 2m{\psi}C'^2 \Phi_{\delta,\gamma}[\Delta]^2 \leq  \Big[2m\psi{C}'^2C_1'^2\lambda_n\Big] \lambda_n \leq \frac{\lambda_n}{\kappa}
\end{eqnarray*}
The last inequality follows from the bound on $\lambda_n$.

%% file: Appendix.tex
\section{Appendix}

\subsection{\textit{Proof of Proposition 1 (main paper)}} 
\begin{proof}
We note that:
\begin{eqnarray*}
\|\Delta\|_2 &\leq& {\|{\Delta}D_y\|_2}  + \|{\Delta}{L}_y\|_2 + {\|{\Delta}\Theta_{yx}\|_2}{} + \|{\Delta}\Theta_{x}\|_2 \leq (3+\gamma)\Phi_{\gamma}(\Delta)
 \end{eqnarray*}
 Furthermore, recall that
\begin{eqnarray*}
 R_{\Sigma^\star}(\mathcal{F}(\Delta)) = \Sigma^\star\Big[\sum_{k = 2}^{\infty}(-\mathcal{F}(\Delta){\Sigma^\star}^{-1})^k\Big].
\end{eqnarray*}
Using this observation and some algebra, we have that:
\begin{eqnarray*}
\Phi_{\gamma}[\mathcal{F}^{\dagger}R_{\Sigma^\star}(\mathcal{F}(\Delta))] \leq m\psi \Big[\sum_{k = 2}^{\infty} (\psi \|\Delta\|_2)^k \Big]
 &\leq& m\psi^3 \frac{(3+\gamma)^2\Phi_{\gamma}[\Delta]^2}{1-(3+\gamma)\Phi_{\gamma}[\Delta]\psi}\\
&\leq& 2m{\psi}C'^2 \Phi_{\gamma}[\Delta]^2
\end{eqnarray*}
\end{proof}

\subsection{\textit{Proof of Proposition 2 (main paper)}} 
\begin{proof}
{\noindent} The proof of this result uses Brouwer's fixed-point theorem, and is inspired by the proof of a similar result in \cite{Ravikumar,Chand2012}.
The optimality conditions of \eqref{eqn:ConvexRelaxed231_N} suggest that there exist Lagrange multipliers $Q_{D_y} \in \mathcal{W}$,  $Q_{T_y} \in {T_y'}^\perp$, and $Q_{T_{yx}} \in {T_{yx}'}^\perp$ such that
\begin{eqnarray*}
[\Sigma_n - {\tilde{\Theta}}^{-1}]_y + Q_{D_y} = 0; \hspace{.1in} [\Sigma_n - {\tilde{\Theta}}^{-1}]_y + Q_{T_y}  &\in& \lambda_n\partial\|\tilde{L}_y\|_\star\\[.01in]
[\Sigma_n - {\tilde{\Theta}}^{-1}]_{yx} + Q_{T_{yx}}  \in -\lambda_n\gamma\partial\|\tilde{\Theta}_{yx}\|_{\star};  \hspace{.1in} [\Sigma_n - {\tilde{\Theta}}^{-1}]_{x}  &=& 0
\end{eqnarray*}
{\noindent}Letting the SVD of $\tilde{L}$ and $\tilde{\Theta}_{yx}$ be given by $\tilde{L}_y = \bar{U}\bar{D}\bar{V}'$ and $\tilde{\Theta}_{yx} = \breve{U}\breve{D}{\breve{V}}'$ respectively, and $Z \triangleq  (0, \hspace{.1in} \lambda_n\bar{U}\bar{V}', \hspace{.1in}  -\lambda_n\gamma_{}{\breve{U}}{\breve{V}}',  \hspace{.1in} 0)$, we can restrict the optimality conditions  of \eqref{eqn:ConvexRelaxed231_N}(main paper) to the space $\mathbb{H}'$ to obtain, $\mathcal{P}_{{\mathbb{H}}'}\mathcal{F}^{\dagger}(\Sigma_n - \tilde{\Theta}^{-1}) = Z$. Further, by appealing to the matrix inversion lemma, this condition can be restated as $\mathcal{P}_{{\mathbb{H}}_{\mathcal{M}}}\mathcal{F}^{\dagger}(E_{n} - R_{\Sigma^\star}(\mathcal{F}\Delta) + \mathbb{I}^{\star}\mathcal{F}(\Delta)) = Z$. Based on the Fisher information assumption 1 in \eqref{eqn:FirstFisherCond} (main paper), the optimum of \eqref{eqn:ConvexRelaxed231_N}(main paper) is unique (this is because the Hessian of the negative log-likelihood term is positive definite restricted to the tangent space constraints). Moreover, using standard Lagrangian duality, one can show that the set of variables  $(\tilde{\Theta}, \tilde{D}_y, \tilde{L}_y)$ that satisfy the restricted optimality conditions are unique. Consider the following function $S(\underline{\delta})$ restriced to $\underline{\delta} \in \mathcal{W} \times T'_y \times T'_{yx} \times \Sp^q$ with $\rho(T(L_y^\star),T_y') \leq \omega_y$ and $\rho(T(\Theta_{yx}^\star),T_{yx}') \leq \omega_{yx}$:
\begin{eqnarray*}
S(\underline{\delta}) = \underline{\delta} - (\mathcal{P}_{\mathbb{H}'_{}}\mathcal{F}^{\dagger}\mathbb{I}^{\star}\mathcal{F}\mathcal{P}_{\mathbb{H}'_{}})^{-1}\Big(\mathcal{P}_{\mathbb{H}'_{}}\mathcal{F}^{\dagger} [E_{n} &-& R_{\Sigma^\star}\mathcal{F}(\underline{\delta} + C_{T_{}}) \\ + \mathbb{I}^{\star}\mathcal{F}(\underline{\delta} + C_{T_{}})] - Z\Big)
\end{eqnarray*}
{\noindent}The function $S(\underline{\delta})$ is well-defined since the operator $\mathcal{P}_{\mathbb{H}'_{}}\mathcal{F}^{\dagger}\mathbb{I}^{\star}\mathcal{F}\mathcal{P}_{\mathbb{H}'_{}}$ is bijective due to Fisher information assumption 1 in \eqref{eqn:FirstFisherCond} (main paper). As a result, $\underline{\delta}$ is a fixed point of $S(\underline{\delta})$ if and only if $\mathcal{P}_{\mathbb{H}'_{}}\mathcal{F}^{\dagger} [E_{n} - R_{\Sigma^\star}(\mathcal{F}(\underline{\delta} + C_{T_{}})) + \mathbb{I}^{\star}\mathcal{F}(\underline{\delta} + C_{T_{}})] =  Z$.
Since the pair $(\tilde{\Theta}, \tilde{D}_y, \tilde{L}_y)$ are the unique solution to \eqref{eqn:ConvexRelaxed231_N}(main paper), the only fixed point of $S$ is $\mathcal{P}_{\mathbb{H}'_{}}[\Delta]$. Next we show that this unique optimum lives inside the ball $\mathbb{B}_{r_1^u,r_2^u} = \{\underline{\delta} \hspace{.1in} | \hspace{.1in} \max\{\|\delta_2\|_2, \frac{1}{\gamma}\|\delta_3\|_2\} \leq r_1^u, \max\{\|\delta_1\|_2, \|\delta_4\|_2\} \leq r_2^u \hspace{.05in} \underline{\delta} \in \mathbb{H}'\}$. In particular, we show that under the map $S$, the image of $\mathbb{B}_{r_1^u,r_2^u}$ lies in $\mathbb{B}_{r_1^u, r_2^u}$ and appeal to Brouwer's fixed point theorem to conclude that $\mathcal{P}_{\mathbb{H}'_{}}[\Delta] \in \mathbb{B}_{r_1^u, r_2^u}$. For $\underline{\delta} \in \mathbb{B}_{r_1^u, r_2^u}$, the first component of $S(\underline{\delta})$, denoted by $S(\underline{\delta})_1$, can be bounded as follows:
\begin{eqnarray*}
\|S(\underline{\delta})_1\|_2 &=& \Big\|\Big[(\mathcal{P}_{\mathbb{H}'_{}}\mathcal{F}^{\dagger}\mathbb{I}^{\star}\mathcal{F}\mathcal{P}_{\mathbb{H}'_{}})^{-1}\Big(\mathcal{P}_{\mathbb{H}'}\mathcal{F}^{\dagger} [E_{n} - R_{\Sigma^\star}(\mathcal{F}(\underline{\delta} + C_{T_{}})) \\&+& \mathbb{I}^{\star}\mathcal{F}C_{T_{}}] + Z\Big)\Big]_1\Big\|_2
\leq \frac{2}{\alpha}\Big[\Phi_{\gamma}[\mathcal{F}^{\dagger}( E_n + \mathbb{I}^{\star}\mathcal{F}(C_{T_{}}))]\Big]\\ &+& \frac{2}{\alpha}{\Phi_{\gamma}[\mathcal{F}^{\dagger}R_{\Sigma^\star}(\underline{\delta}+C_{T})}] \leq \frac{r_2^u}{2} +  \frac{2}{\alpha}{\Phi_{\gamma}[\mathcal{F}^{\dagger}R_{\Sigma^\star}(\underline{\delta}+C_{T})}]
\end{eqnarray*}
The first inequality holds because of Fisher Information Assumption 1 in \eqref{eqn:FirstFisherCond} (main paper), and the properties that $\Phi_{\gamma}[\mathcal{P}_{\mathbb{H}_{\mathcal{M}}}(.)] \leq 2\Phi_{\gamma}(.)$  (since projecting into the tangent space of a low-rank matrix variety increases the spectral norm by a factor of at most two) and $\Phi_{\gamma}(Z) = \lambda_n$. Moreover, since $r_1^u \leq \frac{1}{4C'}$, we have $\Phi_{\gamma}(\underline{\delta} + C_{T_{}}) \leq \Phi_{\gamma}(\underline{\delta}) + \Phi_{\gamma}(C_{T_{}}) \leq 2r_1^u \leq \frac{1}{2C'}$. Moreover, $r_1^u \leq r_2^u\max\{1+\frac{\kappa}{2},\frac{\alpha}{8}\}$. We can now appeal to Proposition 1 (main paper) to obtain:
\begin{eqnarray*}
\frac{2}{\alpha}{\Phi_{\gamma} [\mathcal{F}^{\dagger}R_{\Sigma^\star}(\underline{\delta} + C_{T_{}})}] &\leq& \frac{4}{\alpha}m{\psi}C'^2 [\Phi_{\gamma}(\underline{\delta} + C_{T_{}})]^2 \\ &\leq& \frac{16}{\alpha}m{\psi}C'^2(r_2^u)^2\max\{1+\frac{\kappa}{2},\frac{\alpha}{8}\}^2 \\ &\leq& \frac{r_2^u}{2}
\end{eqnarray*}
Thus, we conclude that $\|S(\delta)_1\|_2 \leq r_2^u$. Similarly, we check that:
\begin{eqnarray*}
\|[S(\underline{\delta})_2]\|_2 &=& \Big\|\Big[(\mathcal{P}_{\mathbb{H}'_{}}\mathcal{F}^{\dagger}\mathbb{I}^{\star}\mathcal{F}\mathcal{P}_{\mathbb{H}'_{}})^{-1}\Big(\mathcal{P}_{\mathbb{H}'}\mathcal{F}^{\dagger} [E_{n} - R_{\Sigma^\star}(\mathcal{F}(\underline{\delta} + C_{T_{}})) \\&+& \mathbb{I}^{\star}\mathcal{F}C_{T_{}}] + Z\Big)\Big]_2
\Big\|_2
\leq \frac{2}{\alpha}\Big[\Phi_{\gamma}[\mathcal{F}^{\dagger}( E_n + \mathbb{I}^{\star}\mathcal{F}(C_{T_{}})] + \lambda_n\Big]\\ &+& \frac{2}{\alpha}{\Phi_{\gamma}[\mathcal{F}^{\dagger}R_{\Sigma^\star}(\underline{\delta}+C_{T})}] \leq \frac{r_1^u}{2} +  \frac{2}{\alpha}{\Phi_{\gamma}[\mathcal{F}^{\dagger}R_{\Sigma^\star}(\underline{\delta}+C_{T})}] \leq r_1^u
\end{eqnarray*}

Using a similar approach, we can conclude that $\frac{1}{\gamma}\|S(\delta)_3\|_2 \leq r_1^u$ and $\|S(\delta)_3\|_2 \leq r_2^u$. Therefore, Brouwer's fixed point theorem suggests that $\mathcal{P}_{\mathbb{H}'}(\Delta) \in \mathcal{B}_{r_1^u, r_2^u}$. Hence, $\|\Delta_1\|_2 \leq r_2^u$,  $\|\Delta_4\|_2 \leq r_2^u$, \\$\|\Delta_2\|_2 \leq \|\mathcal{P}_{\mathbb{H}'[2]}(\Delta_2)\|_2 +  \|\mathcal{P}_{\mathbb{H}'[2]^\perp}(\Delta_2)\|_2 \leq 2r_1^u$, and \\$\frac{1}{\gamma}\|\Delta_3\|_2 \leq \frac{1}{\gamma}\|\mathcal{P}_{\mathbb{H}'[3]}(\Delta_3)\|_2 +  \frac{1}{\gamma}\|\mathcal{P}_{\mathbb{H}'[3]^\perp}(\Delta_2)\|_2 \leq 2r_1^u$.
\end{proof}

\subsection{Proof of Proposition 3 (main paper)}
Below, we outline our proof strategy:
\begin{enumerate}
\item We proceed by analyzing \eqref{eqn:ConvexRelaxed_N} with additional constraints that the variables  ${L}_y$, and $\Theta_{yx}$ belong to the algebraic varieties low-rank matrices (specified by rank of $L_y^\star$, and $\Theta_{yx}^\star$) , and that the tangent spaces $T(L_y)$, $T(\Theta_{yx})$ are close to the nominal tangent spaces $T({L}_y^\star)$, and $T(\Theta_{yx}^\star)$ respectively. We prove that under suitable conditions on the minimum nonzero singular value of ${L}_y^\star$, and minimum nonzero singular value of $\Theta_{yx}^\star$, any optimum pair of variables $(\Theta, D_y, L_y)$ of this non-convex program are smooth points of the underlying varieties; that is $\rm{rank}(L_y) = \rm{rank}(L_y^\star)$ and $\rm{rank}(\Theta_{yx}) = \rm{rank}(\Theta_{yx}^\star)$. Further, we show that $L_y$ has the same inertia as $L_y^\star$ so that $L_y \succeq 0$.
\item Conclusions of the previous step imply the the variety constraints can be ``linearized" at the optimum of the non-convex program to obtain tangent-space constraints. Under the specified conditions on the regularization parameter $\lambda_n$, we prove that with high probability, the unique optimum of this ``linearized" program coincides with the global optimum of the non-convex program.
\item Finally, we show that the tangent-space constraints of the linearized program are inactive at the optimum. Therefore the optimal solution of \eqref{eqn:ConvexRelaxed_N} has the property that with high probability: $\text{rank}(\bar{L}_y) = \text{rank}(L_y^\star)$ and $\text{rank}(\bar{\Theta}_{yx}) = \text{rank}(\Theta_{yx}^\star)$. Since $\bar{L}_y \succeq 0$, we conclude that the variables $(\bar{\Theta}, \bar{D}_y, \bar{L}_y)$ are the unique optimum of \eqref{eqn:main}. 
\end{enumerate}

\subsubsection{Variety Constrained Program}
We begin by considering a variety-constrained optimization program. Letting $(M,N,P,Q) \subset \Sp^p \times \Sp^p \times \R^{p\times{q}} \times \Sp^q$, we denote $\mathcal{P}_{[2,3]}(M,N,P,Q) = (N,P) \subset \Sp^p \times \R^{p\times{q}}$. The variety-constrained optimization program is given by:
\begin{eqnarray}
({\Theta}^{\mathcal{M}}, {D}_y^{\mathcal{M}}, {L}_y^{\mathcal{M}}) = \argmin_{\substack{\Theta \in \Sp^{q+p}, ~\Theta \succ 0 \\ D_y,{L}_y \in \Sp^p}} & -\ell(\Theta;\mathcal{D}^+_n) + \lambda_n [\|{L}_y\|_{\star} + \gamma \|\Theta_{yx}\|_\star] \nonumber \\ \mathrm{s.t.} & \Theta_y = D_y - {L}_y,  (\Theta, D_y,  {L}_y) \in \mathcal{M}. \label{eqn:NConveProblem_N}
\end{eqnarray}
{\noindent}Here, the set $\mathcal{M} = \mathcal{M}_1 \cap \mathcal{M}_2$, where the sets $\mathcal{M}_1$ and $\mathcal{M}_2$ are given by:
\begin{eqnarray*}
\mathcal{M}_1 &\triangleq& \Big\{(\Theta, D_y, {L}_y) \in \Sp^{(p+q)} \times \Sp^p \times \Sp^{p} \Big | D_y \text{ is diagonal}, \hspace{.1in} \text{rank}({L}_y) \leq \text{rank}({L}_y^\star) \\ & \text{rank}&(\Theta_{yx}) \leq \text{rank}(\Theta_{yx}^\star); \|\mathcal{P}_{T({L}_y^\star)^{\perp}}({L}_y - {L}_y^\star)\|_2 \leq \frac{\lambda_n}{2\psi^2} \\
& &\|\mathcal{P}_{T(\Theta_{yx}^\star)^{\perp}}(\Theta_{yx} - \Theta_{yx}^\star)\|_2 \leq \frac{\lambda_n}{2\psi^2}\Big\} \\
 \mathcal{M}_2 &\triangleq& \Big\{(\Theta, D_y, {L}_y) \in \Sp^{(p+q)} \times \Sp^p \times \Sp^{p} \Big |
\hspace{.1in} \\ & & \Gamma_{\gamma}[\mathcal{P}_{[2,3]}(\mathbb{I}^{\star}\mathcal{F}(\Delta))] \leq   \frac{\lambda_n}{\kappa}\Big(20m(1+\frac{4}{\alpha})\psi^2+\kappa+13\Big) \\ & & \|\mathbb{I}^\star\mathcal{F}(\Delta_1,0,0,\Delta_4)\|_2 \leq \frac{8\lambda_n}{\kappa}\Big(1+\frac{4}{\alpha}\Big)\psi^2 \Big\},
\end{eqnarray*}

{\noindent}The optimization program \eqref{eqn:NConveProblem_N} is non-convex due to the rank constraints $\text{rank}(L_y) \leq \text{rank}({L}_y^\star)$ and $\text{rank}(\Theta_{yx}) \leq \text{rank}(\Theta_{yx}^\star)$ in the set $\mathcal{M}$. These constraints ensure that the matrices ${L}_y$, and $\Theta_{yx}$ belong to appropriate varieties. The constraints in $\mathcal{M}$ along $T(L_y^\star)^\perp$ and $T(\Theta_{yx}^\star)^\perp$ ensure that the tangent spaces $T({L}_y)$ and $T(\Theta_{yx})$ are ``close'' to $T({L}_y^\star)$ and $T(\Theta_{yx}^\star)$ respectively. Finally, the last conditions roughly controls the error. We begin by proving the following useful proposition:

\begin{proposition}
\label{prop:FirstResultCor}
Let $(\Theta, D_y, {L}_y)$ be a set of feasible variables of \eqref{eqn:NConveProblem_N}. Let $\Delta = (D_y - D_y^\star, {L}_y - {L}_y^\star, \Theta_{yx} - \Theta_{yx}^\star, \Theta_{x} - \Theta_{x}^\star)$ and recall that $C_1' =  \frac{2m}{\alpha\kappa}\Big(28(1+\frac{4}{\alpha})\psi^2+2\kappa+13\Big)+ \frac{1}{\psi^2}$. Then, $\Phi_{\gamma}[\Delta] \leq C_1'\lambda_n$
\label{theorem:NonconvexImp}
\end{proposition}
\begin{proof}
Let ${\mathbb{H}}^\star = \mathcal{W} \times T({L}_y^\star) \times T(\Theta_{yx}^\star) \times \Sp^{q}$. Then,
\begin{eqnarray*}
\Phi_{\gamma}[\mathcal{F}^{\dagger}\mathbb{I}^{\star}\mathcal{F}\mathcal{P}_{\mathbb{H}^\star}(\Delta)] &\leq& \Phi_{\gamma}[\mathcal{F}^{\dagger}\mathbb{I}^{\star}\mathcal{F}(\Delta_{})] + \Phi_{\gamma}[\mathcal{F}^{\dagger}\mathbb{I}^{\star}\mathcal{F}\mathcal{P}_{{\mathbb{H}^\star}^{\perp}}(\Delta)] \\
&\leq& \frac{8m\lambda_n}{\kappa}\Big(1+\frac{4}{\alpha}\Big)\psi^2 + \frac{\lambda_n}{\kappa}\Big(20m(1+\frac{4}{\alpha})\psi^2+\kappa+13\Big) \\&+& m\psi^2\Big(\frac{\omega_{y}\lambda_n}{2\psi^2} + \frac{\omega_{yx}\lambda_n}{2\psi^2}\Big) \\ 
&\leq& \frac{m\lambda_n}{\kappa}\Big(28(1+\frac{4}{\alpha})\psi^2+2\kappa+13\Big)
\end{eqnarray*}
Since $\Phi_{\gamma}[\mathcal{P}_{\mathbb{H}^\star}(\cdot)] \leq 2\Phi_{\gamma}(\cdot)$, we have that $\Phi_{\gamma}[\mathcal{P}_{\mathbb{H}^\star}\mathcal{F}^{\dagger}\mathbb{I}^{\star}\mathcal{F}\mathcal{P}_{\mathbb{H}^\star}(\Delta)] \leq \frac{2m\lambda_n}{\kappa}\Big(28(1+\frac{4}{\alpha})\psi^2+2\kappa+13\Big)$. Consequently, we apply Fisher Information Assumption 1 in \eqref{eqn:FirstFisherCond} (main paper) to conclude that $\Phi_{\gamma}[\mathcal{P}_{\mathbb{H}^\star}(\Delta)] \leq \frac{2m\lambda_n}{\alpha\kappa}\Big(28(1+\frac{4}{\alpha})\psi^2+2\kappa+13\Big)$. Moreover:
\begin{eqnarray*}
\Phi_{\gamma}[\Delta] \leq \Phi_{\gamma}[\mathcal{P}_{\mathbb{H}^\star}(\Delta_{})] + \Phi_{\gamma}[\mathcal{P}_{{\mathbb{H}^\star}^{\perp}}(\Delta)] &\leq& \frac{2m\lambda_n}{\alpha\kappa}\Big(28(1+\frac{4}{\alpha})\psi^2+2\kappa+13\Big)+ \frac{\lambda_n}{\psi^2}
\\ &=& C_1'\lambda_n
\end{eqnarray*}
\end{proof}
Proposition~\ref{theorem:NonconvexImp} leads to powerful implications. In particular, under additional conditions on the minimum nonzero singular values of ${L}_y^\star$ and $\Theta_{yx}^\star$, any feasible set of variables $(\Theta, D_y, {L}_y)$ of \eqref{eqn:NConveProblem_N} has two key properties: $(a)$ The variables $(\Theta_{yx}, {L}_y)$ are smooth points of the underlying varieties, $(b)$ The constraints in $\mathcal{M}$ along $T({L}_y^\star)^{\perp}$ and $T(\Theta_{yx}^\star)^{\perp}$ are locally inactive at $\Theta_{yx}$ and $L_y$. These properties, among others, are proved in the following corollary.

\begin{corollary}
\label{eqn:Corollary1}
Consider any feasible variables $( \Theta, D_y, {L}_y)$ of \eqref{eqn:NConveProblem_N}. Let $\sigma_y$ be the smallest nonzero singular value of ${L}_y^\star$ and $\sigma_{yx}$ be the smallest nonzero singular value of $\Theta_{yx}^\star$. Let $\mathbb{H}' = \mathcal{W} \times T({L}_y) \times T(\Theta_{yx}) \times \Sp^q$ and $C_{T'} = \mathcal{P}_{\mathbb{H}'^{\perp}} (0, {L}_y^\star, \Theta_{yx}^\star ,0)$. Furthermore, recall that $C_1' =  \frac{2m}{\alpha\kappa}\Big(28(1+\frac{4}{\alpha})\psi^2+2\kappa+13\Big)+ \frac{1}{\psi^2}$, $C_2' = \frac{4}{\alpha_{}} (1+\frac{2}{\kappa})$, $C_{\sigma_y}' = C_1'^2\psi^2\max\{2\kappa + 1, \frac{2}{C_2'\psi^2} + 1\}$ and $C_{\sigma_{yx}}' = C_1'^2\psi^2\max\{2\kappa + \frac{\kappa}{\gamma}, \frac{2}{C_2'\psi^2} + \frac{\kappa}{\gamma}\}$. Suppose that the following inequalities are met: $\sigma_y \geq  \frac{{m}}{\omega_y}C_{\sigma_y}\lambda_n$, \\ $\sigma_{yx} \geq  \frac{m\gamma^2}{\omega_{yx}}C_{\sigma_{yx}}'\lambda_n$. Then,
\begin{enumerate}
\item ${L}_y$ and $\Theta_{yx}$ are smooth points of their underlying varieties, i.e. $\rm{rank}({L}_y) = \rm{rank}({L}_y^\star)$, $\rm{rank}(\Theta_{yx}) = \rm{rank}(\Theta_{yx}^\star)$; Moreover ${L}_y$ has the same inertia as ${L}_y^\star$.
\item $\|\mathcal{P}_{T({L}_y^\star)^{\perp}}({L}_y - {L}_y^\star)\|_2 \leq \frac{\lambda_n\omega_{y}}{48m\psi^2}$ and $\|\mathcal{P}_{T(\Theta_{yx}^\star)^{\perp}}(\Theta_{yx} - \Theta_{yx}^\star)\|_2 \leq \frac{\lambda_n\omega_{yx}}{48m\psi^2}$
\item $\rho(T({L}_y), T({L}_y^\star)) \leq \omega_{y}$; $\rho(T({\Theta}_{yx}), T({\Theta}_{yx}^\star)) \leq \omega_{yx}$; that is, the tangent spaces at $L_y$ and $\Theta_{yx}$ is ``close" to the tangent space $L_y^\star$ and $\Theta_{yx}^\star$. 
\item $\Phi_{\gamma}[C_{T_{}'}] \leq \min\{\frac{\lambda_n}{\kappa\psi^2},C_2'\lambda_n\}$
\end{enumerate}
\end{corollary}

\begin{proof}
We note the following relations before proving each step: $C_1' \geq \frac{1}{\psi^2} \geq \frac{1}{m\psi^2}$, $\omega_{y}, \omega_{yx} \in (0,1)$, and $\kappa \geq 6$. We also appeal to the results of  regarding perturbation analysis of the low-rank matrix variety \cite{Bach2008}.\\

1. Based on the assumptions regarding the minimum nonzero singular values of ${L}_y^\star$ and $\Theta_{yx}^\star$, one can check that:
\begin{eqnarray*}
\sigma_y &\geq& \frac{ C_1'^2\lambda_n}{\omega_{y}} m\psi^2({\kappa+1})
\geq \frac{ C_1'\lambda_n}{\omega_{y}}({2\kappa+1}) \geq 8\|L - L_y^\star\|_2\\
\sigma_{yx} &\geq& \frac{C_1'^2\lambda_n}{{\omega_{yx}}} \gamma^2m\psi^2{\Big(\frac{6\beta}{\gamma}+2\kappa\Big)}
\geq  8\|{\Theta}_{yx} - \Theta_{yx}^\star\|_2
\end{eqnarray*}

{\noindent}Combining these results and Proposition~\ref{theorem:NonconvexImp}, we conclude that ${L}_y$ and $\Theta_{yx}$ are smooth points of their respective varieties, i.e. $\text{rank}({L}_y)= \text{rank}({L}_y^\star)$, and $\text{rank}(\Theta_{yx}) = \text{rank}(\Theta_{yx}^\star)$. Furthermore, ${L}_y$ has the same inertia as ${L}_y^\star$.

2. Since $\sigma_y \geq 8\|L_y - L_y^\star\|_2$, and $\sigma_{yx} \geq 8\|\Theta_{yx} - \Theta_{yx}^\star\|_2$, we can appeal to Proposition 2.2 of \cite{Chand2012} to conclude that the constraints in $\mathcal{M}$ along $\mathcal{P}_{T({L}_y^\star)^{\perp}}$ and $\mathcal{P}_{T(\Theta_{yx}^\star)^{\perp}}$ are strictly feasible:
\begin{eqnarray*}
\|\mathcal{P}_{T({L}_y^\star)^{\perp}}({L}_y - {L}_y^\star)\|_2 &\leq& \frac{\|{L}_y - {L}_y^\star\|_2^2}{\sigma_y}
\leq \frac{\lambda_n}{48m\psi^2} \\
\|\mathcal{P}_{T(\Theta_{yx}^\star)^{\perp}}(\Theta_{yx} - \Theta_{yx}^\star)\|_2 &\leq& \frac{\|\Theta_{yx} - \Theta_{yx}^\star\|_2^2}{\sigma_{yx}}
\leq \frac{\lambda_n}{48m\psi^2}
\end{eqnarray*}

3. Appealing to Proposition 2.1 of \cite{Chand2012}, we prove that the tangent spaces $T(L_y)$ and $T(\Theta_{yx})$ are close to $T(L_y^\star)$ and $T(\Theta_{yx}^\star)$ respectively:
\begin{eqnarray*}
\rho(T({L}_y), T({L}_y^\star)) &\leq& \frac{2\| {L}_y - {L}_y^\star \|_2}{\sigma_y}
\leq \frac{2C_1'\lambda_n \omega_{}}{C_1'^2\lambda_nm{\psi}^2(2\kappa+1)}
\leq \omega_{y}\\
\rho(T({\Theta}_{yx}), T({\Theta}_{yx}^\star)) &\leq& \frac{2\| {\Theta}_{yx} - {\Theta}_{yx}^\star \|_2}{\sigma_{yx}}
\leq \frac{2C_1'\lambda_n\gamma \omega_{yx}}{\frac{C_1'^2\lambda_n}{{\omega_{yx}}} \gamma^2m\psi^2{\Big(\frac{\kappa}{\gamma}+2\kappa\Big)}}
\leq \omega_{yx}\\
\end{eqnarray*}

4. Letting $\sigma_y'$ and $\sigma_{yx}'$ be the minimum nonzero singular value of ${L}_y$ and $\Theta_{yx}$ respectively, one can check that:
\begin{eqnarray*}
\sigma_y' \geq \sigma_y - \|{L}_y - {L}_y^\star\|_2 \geq 8C_1'\lambda_n \geq 8\|{L}_y - L_y^\star\|_2 \\[.1in]
\sigma_{yx}' \geq \sigma_{yx} - \|\Theta_{yx} - \Theta_{yx}^\star\|_2 \geq 8C_1'\lambda_n\gamma \geq 8\|\Theta_{yx} - {\Theta}_{yx}^\star\|_2
\end{eqnarray*}
Once again appealing to Proposition 2.2 of \cite{Chand2012} and simple algebra, we have:
\begin{eqnarray*}
\Phi_{\gamma}(C_{T'}) &\leq& m \|\mathcal{P}_{T({L_y})^{\perp}} ({L}_y - {L}_y^\star)\|_2 + m \|\mathcal{P}_{T(\Theta_{yx})^{\perp}} (\Theta_{yx} - \Theta_{yx}^\star)\|_2 \\ &\leq& m\frac{\|{L}_y - {L}_y^\star\|_2^2}{\sigma_y'} + m\frac{\|\Theta_{yx} - \Theta_{yx}^\star\|_2^2}{\sigma_{yx}'} \leq \min \Big\{\frac{\lambda_n}{\kappa\psi^2}, C_2'\lambda_n\Big\}
\end{eqnarray*}
\end{proof}

\subsubsection{Variety Constrained Program to Tangent Space Constrained Program}
Consider any optimal solution $(\Theta^{\mathcal{M}}, D_y^{\mathcal{M}}, {L}_y^{\mathcal{M}})$ of \eqref{eqn:NConveProblem_N}. In Corollary~\ref{eqn:Corollary1}, we concluded that the variables $(\Theta^{\mathcal{M}}_{yx}, {L}_y^{\mathcal{M}})$ are smooth points of their respective varieties. As a result, the rank constraints  $\rm{rank}({L}_y) \leq \rm{rank}({L}_y^\star)$ and $\rm{rank}(\Theta_{yx}) \leq \rm{rank}(\Theta_{yx}^\star)$ can be ``linearized" to ${L}_y \in T(L_y^{\mathcal{M}})$ and $\Theta_{yx} \in T(\Theta^{\mathcal{M}}_{yx})$ respectively. Since all the remaining constraints are convex, the optimum of this linearized program is also the optimum of \eqref{eqn:NConveProblem_N}. Moreover, we once more appeal to Corollary~\ref{eqn:Corollary1} to conclude that the constraints in $\mathcal{M}$ along $\mathcal{P}_{T({L}_y^\star)^{\perp}}$ and $\mathcal{P}_{T(\Theta_{yx}^\star)^{\perp}}$ are strictly feasible at $(\Theta^{\mathcal{M}}, D_y^{\mathcal{M}}, {L}_y^{\mathcal{M}})$. As a result, these constraints are locally inactive and can be removed without changing the optimum. Therefore the constraint $(\Theta^{\mathcal{M}}, D_y^{\mathcal{M}}, {L}_y^{\mathcal{M}}) \in \M_1$ is inactive and can be removed. We now argue that the constraint $(\Theta^{\mathcal{M}}, D_y^{\mathcal{M}}, {L}_y^{\mathcal{M}}) \in \M_2$ in \eqref{eqn:NConveProblem_N} can also removed in this ``linearized" convex program. In particular, letting $\mathbb{H}_{\mathcal{M}} \triangleq \mathcal{W} \times T({L}_y^{\mathcal{M}}) \times T(\Theta^{\mathcal{M}}_{yx}) \times \Sp^q$, consider the following convex optimization program:
\begin{eqnarray}
(\tilde{\Theta}, \tilde{D}_y, \tilde{L}_y) = \argmin_{\substack{\Theta \in \Sp^{q+p}, ~\Theta \succ 0 \\ D_y,{L}_y \in \Sp^p}} & -\ell(\Theta; \mathcal{D}^+_n) + \lambda_n [\|{L_y}\|_{\star} + \gamma \|\Theta_{yx}\|_\star] \nonumber \\ \mathrm{s.t. } & \hspace{-0.021in} \Theta_y = D_y - {L}_y, ~ (D_y, {L}_y, \Theta_{yx}, \Theta_{x}) \in \mathbb{H}_{\mathcal{M}} \label{eqn:ConvexRelaxed21_N}
\end{eqnarray}

{\noindent}We prove that under conditions imposed on the regularization parameter $\lambda_n$, the pair of variables $(\Theta^{\mathcal{M}}, D_y^{\mathcal{M}}, {L}_y^{\mathcal{M}})$ is the unique optimum of \eqref{eqn:ConvexRelaxed21_N}. That is, we show that
\begin{enumerate}
\item $\Gamma_{\gamma}[\mathcal{P}_{[2,3]}(\mathbb{I}^{\star}\mathcal{F}(\Delta))] < \frac{2\lambda_n}{\kappa}\Big(20m(1+\frac{4}{\alpha})\psi^2+\kappa+13\Big)$\\[0.05in] 
\item $\|\mathbb{I}^\star\mathcal{F}(D_y^{\mathcal{M}}-D_y^\star,0,0,\Theta_{x}^\M-\Theta_{x}^\star)\|_2 < \frac{8\lambda_n}{\kappa}\Big(1+\frac{4}{\alpha}\Big)\psi^2$
\end{enumerate}

Appealing to Corollary~\ref{eqn:Corollary1} and Proposition~\ref{prop:Enbound}, we have that $\Phi_{\gamma}[\mathcal{F}^{\dagger}\mathbb{I}^{\star}\mathcal{F}C_{T_{\mathcal{M}}}]\leq \frac{\lambda_n}{\kappa}$, $\Phi_{\gamma}[C_{T_{\mathcal{M}}}] \leq C_2'\lambda_n$ and  (with high probability) $\Phi_{\gamma}[\mathcal{F}^{\dagger}E_{n}] \leq \frac{\lambda_n}{\kappa}$. Consequently, based on the bound on $\lambda_n$ in assumption of Theorem~\ref{proposition:lambdaBound}, it is straightforward to show that $r_1^u \leq \min\Big\{\frac{1}{4C'}, \frac{\alpha}{32\max\{1+\frac{\kappa}{2},\frac{\alpha}{8}\}^2m{\psi}C'^2}\Big\}$. We further have that $r_2^u \leq \frac{2\lambda_n}{\kappa}\Big(1+\frac{4}{\alpha}\Big)$. Hence by Proposition 2 (main paper), $\|D_y^{\mathcal{M}}-D_y^\star\|_2 \leq r_2^u$ and $\|\Theta_x^{\mathcal{M}}-\Theta_x^\star\|_2 \leq r_2^u$. Thus, $\|\mathbb{I}^\star\mathcal{F}(D_y^{\mathcal{M}}-D_y^\star,0,0,\Theta_{x}^\M-\Theta_{x}^\star)\|_2 \leq 2\psi^2r_2^u < \frac{8\lambda_n}{\kappa}\Big(1+\frac{4}{\alpha}\Big)\psi^2$. From Proposition 2 (main paper), we also have that $\Phi_{\gamma}[\Delta] \leq \frac{1}{2C'}$. Finally, we can appeal to Proposition 1 (main paper) and the bound on $\lambda_n$ to conclude $\Phi_{\gamma}[\mathcal{F}^{\dagger}R_{\Sigma^\star} (\mathcal{F}(\Delta_{}))] \leq {2m{\psi}C'^2 \Phi_{\gamma}[\Delta]^2} \leq {2m{\psi}C'^2C_1'^2\lambda_n^2} \leq \frac{\lambda_n}{\kappa}$. Based on the optimality condition of \eqref{eqn:ConvexRelaxed21_N}, the property that $\Phi_{\gamma}[\mathcal{P}_{\mathbb{H}_{\mathcal{M}}}( . )] \leq 2\Phi_{\gamma}(.)$, we have:
\begin{eqnarray*}
\Phi_{\gamma}[\mathcal{P}_{\mathbb{H}_{\mathcal{M}}}\mathcal{F}^{\dagger}\mathbb{I}^{\star}\mathcal{F}\mathcal{P}_{\mathbb{H}_{\mathcal{M}}}(\Delta)] &\leq& 2\lambda_n + \Phi_{\gamma}[\mathcal{P}_{\mathbb{H}_{\mathcal{M}}}\mathcal{F}^{\dagger}R_{\Sigma^\star}(\Delta)] + \Phi_{\gamma} [\mathcal{P}_{\mathbb{H}_{\mathcal{M}}} \mathcal{F}^{\dagger}\mathbb{I}^{\star}\mathcal{F}C_{T_{\mathcal{M}}}]\nonumber\\ &+& \Phi_{\gamma}[\mathcal{P}_{\mathbb{H}_{\mathcal{M}}} \mathcal{F}^{\dagger}E_{n}] \\
 &\leq& 2\lambda_n +  2\lambda_n\Big(\frac{1}{\kappa} + \frac{1}{\kappa} + \frac{1}{\kappa}\Big) \leq \frac{2\lambda_n\kappa + 6\lambda_n}{\kappa} 
\end{eqnarray*}
Recalling the notation $\mathbb{H}_\M[2,3] = L_y^\M \times \Theta_{yx}^\M$, we use the result above to conclude that:\\
\begin{eqnarray*}
\|\mathcal{P}_{\mathbb{H}_\M[2]}
\mathcal{G}^\dagger\mathbb{I}^\star\mathcal{G}\mathcal{P}_{\mathbb{H}_\mathcal{M}[2,3]}(\Delta)\|_2 &\leq& \frac{2\lambda_n\kappa + 6\lambda_n}{\kappa} +4\psi^2r_2^u \\ &\leq& \frac{8\lambda_n}{\kappa}\Big(1+\frac{4}{\alpha}\Big)\psi^2 +  \frac{2\lambda_n\kappa + 6\lambda_n}{\kappa}  \\
\frac{1}{\gamma}\|\mathcal{P}_{\mathbb{H}_\M[3]}
F^\dagger\mathbb{I}^\star\F\mathcal{P}_{\mathbb{H}_\mathcal{M}[2,3]}(\Delta)\|_2 &\leq& \frac{2\lambda_n\kappa + 3\lambda_n}{\kappa} +\frac{4\psi^2r_2^u}{\gamma}\\ &\leq& \frac{1}{\gamma}\frac{8\lambda_n}{\kappa}\Big(1+\frac{4}{\alpha}\Big)\psi^2 +  \frac{2\lambda_n\kappa + 6\lambda_n}{\kappa}
\end{eqnarray*}

Thus, 
\begin{eqnarray*}
\Gamma_{\gamma}[\mathcal{P}_{\mathbb{H}_\mathcal{M}[2,3]}\mathcal{G}^\dagger\mathbb{I}^\star\mathcal{G}\mathcal{P}_{\mathbb{H}_\mathcal{M}[2,3]}(\Delta)] &\leq& \frac{4\lambda_n{}}{\kappa}\Big(1+\frac{4}{\alpha}\Big)\psi^2 +  \frac{2\lambda_n\kappa{} + 6\lambda_n{}}{\kappa} \\
&=& \frac{\lambda_n{}}{\kappa}\Big(8m(1+\frac{4}{\alpha})\psi^2+2\kappa+6\Big)
\end{eqnarray*}

Furthermore, observing that the Fisher Information Assumption 3 in \eqref{eqn:SecondFisherCond} (main paper) implies that \\ $\Gamma_{\gamma}[\mathcal{P}_{\mathbb{H}_\mathcal{M}[2,3]^\perp} \mathcal{G}^\dagger\mathbb{I}^\star\mathcal{G}\mathcal{P}_{\mathbb{H}_\mathcal{M}[2,3]}(\Delta)] \leq \Gamma_{\gamma}[\mathcal{P}_{\mathbb{H}_\mathcal{M}[2,3]} \mathcal{G}^\dagger\mathbb{I}^\star\mathcal{G}\mathcal{P}_{\mathbb{H}_\mathcal{M}[2,3]}(\Delta)]$, we have:

\begin{eqnarray*}
\Gamma_{\gamma}(\mathcal{P}_{[2,3]}[\mathcal{G}^\dagger\mathbb{I}^{\star}\mathcal{F}(\Delta)]) &\leq& {2mr_2^u\psi^2}+ \Gamma_{\gamma}[\mathcal{P}_{\mathbb{H}_\mathcal{M}[2,3]} \mathcal{G}^\dagger\mathbb{I}^\star\mathcal{G}\mathcal{P}_{\mathbb{H}_\mathcal{M}[2,3]}(\Delta)]\\ &+&\Gamma_{\gamma}[\mathcal{P}_{\mathbb{H}_\mathcal{M}[2,3]^\perp} \mathcal{F}^\dagger\mathbb{I}^\star\mathcal{F}\mathcal{P}_{\mathbb{H}_\mathcal{M}[2,3]}(\Delta)] + \Phi_{\gamma}[\mathcal{F}^{\dagger}\mathbb{I}^\star\mathcal{F}C_{T_{\mathcal{M}}}] \\ &\leq&\frac{\lambda_n}{\kappa}\Big(20m(1+\frac{4}{\alpha})\psi^2+\kappa+13\Big) \\
&<& \frac{2\lambda_n}{\kappa}\Big(20m(1+\frac{4}{\alpha})\psi^2+\kappa+13\Big) 
\end{eqnarray*}

\subsubsection{From Tangent Space Constraints to the Original Problem}

The optimality conditions of \eqref{eqn:ConvexRelaxed21_N} suggest that there exist Lagrange multipliers $Q_{D_y} \in \mathcal{W}$,  $Q_{T_y} \in T(L_y^\mathcal{M})^\perp$, and $Q_{T_{yx}} \in T(\Theta_{yx}^\mathcal{M})^{\perp}$ such that
\begin{eqnarray*}
[\Sigma_n - {\tilde{\Theta}}^{-1}]_y + Q_{D_y} = 0; \hspace{.1in} [\Sigma_n - {\tilde{\Theta}}^{-1}]_y + Q_{T_y}  &\in& \lambda_n\partial\|\tilde{L}_y\|_\star\\[.01in]
[\Sigma_n - {\tilde{\Theta}}^{-1}]_{yx} + Q_{T_{yx}}  \in -\lambda_n\gamma\partial\|\tilde{\Theta}_{yx}\|_{\star};  \hspace{.1in} [\Sigma_n - {\tilde{\Theta}}^{-1}]_{x}  &=& 0
\end{eqnarray*}
Letting the SVD of $\tilde{L}_y$ and $\tilde{\Theta}_{yx}$ be given by $\tilde{L}_y = \bar{U}\bar{O}\bar{V}'$ and $\tilde{\Theta}_{yx} = \breve{U}\breve{O}{\breve{V}}'$ respectively, and $Z \triangleq  (0, \hspace{.1in} \lambda_n\bar{U}\bar{V}', \hspace{.1in}  -\lambda_n\gamma_{}{\breve{U}}{\breve{V}}',  \hspace{.1in} 0)$, we can restrict the optimality conditions to the space $\mathbb{H}_\mathcal{M}$ to obtain, $\mathcal{P}_{{\mathbb{H}}_\mathcal{M}}\mathcal{F}^{\dagger}(\Sigma_n - \tilde{\Theta}^{-1}) = Z$.
We proceed by proving that the variables $(\tilde{\Theta}, \tilde{D}_y, \tilde{L}_y)$ satisfy the optimality conditions of the original convex program \eqref{eqn:main}. That is:
\begin{center}
\begin{enumerate}
\item $\mathcal{P}_{\mathbb{H}_\mathcal{M}} \mathcal{F}^{\dagger}(\Sigma_n - \tilde{\Theta}^{-1}) = Z$
\item $\max \Big\{\|\mathcal{P}_{T_y'^\perp} (\Sigma_n - \tilde{\Theta}^{-1})_y\|_2, \frac{1}{\gamma}{ \|\mathcal{P}_{T_{yx}'^{\perp}} (\Sigma_n - \tilde{\Theta}^{-1})_{yx}\|_2}\Big\} < \lambda_n$
\end{enumerate}
\end{center}

{\noindent}It is clear that the first condition is satisfied since the pair $(\tilde{\Theta}, \tilde{S}_y, \tilde{L}_y)$ is optimum for \eqref{eqn:ConvexRelaxed21_N}. To prove that the second condition, we must prove that $\Gamma_{\gamma} [\mathcal{P}_{\mathbb{H}_{\mathcal{M}}^{\perp}[2,3]} \mathcal{G}^{\dagger}(\Sigma_n - \tilde{\Theta}^{-1})] < \lambda_n$. In particular, denoting $\Delta = (\tilde{D}_y - D_y^\star, \tilde{L}_y - L_y^\star, \tilde{\Theta}_{yx} - \Theta_{yx}^\star, \tilde{\Theta}_{x} - \Theta_x^\star)$ we show that:
\begin{eqnarray}
\Gamma_{\gamma}[\mathcal{P}_{\mathbb{H}_{\mathcal{M}}^{\perp}[2,3]} \mathcal{G}^{\dagger}\mathbb{I}^{\star}\mathcal{G} \mathcal{P}_{\mathbb{H}_\mathcal{M}[2,3]}(\Delta)]  &<& \lambda_n - \Phi_{\gamma}[\mathcal{P}_{\mathbb{H}_\mathcal{M}^{\perp}}\mathcal{F}^{\dagger} E_n] \\ \nonumber &-& \Phi_{\gamma}[\mathcal{P}_{\mathbb{H}_\mathcal{M}^{\perp}} \mathcal{F}^{\dagger} R_{\Sigma^\star}(\mathcal{F}(\Delta))] \\ \nonumber &-& \Phi_{\gamma}[\mathcal{P}_{\mathbb{H}_\mathcal{M}^{\perp}} \mathcal{F}^{\dagger} \mathbb{I}^{\star}\mathcal{F}C_{T_\mathcal{M}}]\\ \nonumber&-& \Gamma_{\gamma}[\mathcal{P}_{\mathbb{H}_\mathcal{M}[2,3]^{\perp}}\mathcal{G}^{\dagger}\mathbb{I}^\star\mathcal{F}(\Delta_1, 0, 0, \Delta_4)] \nonumber
 \label{eqn:result1}
\end{eqnarray}

Using the fact that $\Gamma_{\gamma}[\mathcal{P}_{{\mathbb{H}_\M[2,3]^\perp}}\mathcal{G}^\dagger(N)] \leq \Phi_{\gamma}[\mathcal{P}_{{\mathbb{H}}^\perp_\M}\mathcal{F}^\dagger(N)]$ for any matrix $N \in \Sp^{p+q}$, this would in turn imply that:
\begin{eqnarray}
\Gamma_{\gamma}[\mathcal{P}_{\mathbb{H}_{\mathcal{M}}[2,3]^\perp} \mathcal{G}^{\dagger}\mathbb{I}^{\star}\mathcal{G} \mathcal{P}_{\mathbb{H}_\mathcal{M}[2,3]}(\Delta)]  &<& \lambda_n - \Gamma_{\gamma}[\mathcal{P}_{\mathbb{H}_\mathcal{M}[2,3]^{\perp}}\mathcal{G}^{\dagger} E_n] \\ \nonumber &-& \Gamma_{\gamma}[\mathcal{P}_{\mathbb{H}_\mathcal{M}[2,3]^{\perp}} \mathcal{G}^{\dagger} R_{\Sigma^\star}(\mathcal{F}(\Delta))] \\ \nonumber&-& \Gamma_{\gamma}[\mathcal{P}_{\mathbb{H}_\mathcal{M}[2,3]^{\perp}} \mathcal{G}^{\dagger} \mathbb{I}^{\star}\mathcal{F}C_{T_\mathcal{M}}]\\ \nonumber &-& \Gamma_{\gamma}[\mathcal{P}_{\mathbb{H}_\mathcal{M}[2,3]^{\perp}}\mathcal{G}^{\dagger}\mathbb{I}^\star\mathcal{F}(\Delta_1, 0, 0, \Delta_4)]\nonumber
 \label{eqn:result2}
\end{eqnarray}

{\noindent}Indeed (26) implies that the second optimality condition is satisfied. So we focus on showing that (27) is satisfied. Using the first optimality condition, the fact that projecting into tangent spaces with respect to rank variety increase the spectral norm by at most a factor of two (i.e. $\Phi_{\gamma}[\mathcal{P}_{\mathbb{H}_\M}(\cdot)] \leq 2\Phi_{\gamma}(\cdot)$),the fact that $\Gamma_{\gamma}[\mathcal{G}^\dagger(\cdot)] \leq \Phi_{\gamma}[\mathcal{F}^\dagger(\cdot)]$, and that $\kappa = \beta(6+\frac{16\psi^2m}{\alpha})$, we have that:
\begin{eqnarray*}
\Gamma_{\gamma}[\mathcal{P}_{\mathbb{H}_\mathcal{M}[2,3]}\mathcal{G}^{\dagger}\mathbb{I}^\star\mathcal{G}\mathcal{}\mathcal{P}_{\mathbb{H}_\mathcal{M}[2,3]}(\Delta)]  &\leq& \lambda_n+ 2\Gamma_{\gamma}[\mathcal{G}^{\dagger}R_{\Sigma^\star}(\Delta)] + 2\Gamma_{\gamma}[\mathcal{G}^{\dagger}\mathbb{I}^{\star} \mathcal{F}C_{T_{M}}]\\ &+& 2\Gamma_{\gamma}[\mathcal{G}^{\dagger}E_{n}]+  \Gamma_{\gamma}[\mathcal{G}^{\dagger}\mathbb{I}^\star\mathcal{F}(\Delta_1, 0, 0, \Delta_4)] \\ &\leq&  \lambda_n+ 2\Phi_{\gamma}[\mathcal{F}^{\dagger}R_{\Sigma^\star}(\Delta)] + 2\Phi_{\gamma}[\mathcal{F}^{\dagger}\mathbb{I}^{\star} \mathcal{F}C_{T_{M}}]\\ &+& 2\Phi_{\gamma}[\mathcal{F}^{\dagger}E_{n}]+  \Phi_{\gamma}[\mathcal{F}^{\dagger}\mathbb{I}^\star\mathcal{F}(\Delta_1, 0, 0, \Delta_4)] \\
&\leq& \lambda_n^{} + \frac{\lambda_n}{\beta}
\end{eqnarray*}

{\noindent}Applying Fisher Information Assumption 2 in \eqref{eqn:SecondFisherCond} (main paper), we obtain:
\begin{eqnarray*}
\Gamma_{\gamma}[\mathcal{P}_{\mathbb{H}_\mathcal{M}[2,3]^{\perp}} \mathcal{G}^{\dagger}\mathbb{I}^{\star}\mathcal{G} \mathcal{P}_{\mathbb{H}_\mathcal{M}[2,3]}(\Delta)] &\leq& \frac{(\beta+1)\lambda_n}{\beta}\Big(1 - \frac{2}{\beta+1}\Big) = \lambda_n - \frac{\lambda_n}{\beta} \\ &<& \lambda_n - \frac{\lambda_n}{2\beta}\\
&\leq& \lambda_n - \Phi_{\gamma}[\mathcal{F}^{\dagger}R_{\Sigma}(\mathcal{F}(\Delta))] - \Phi_{\gamma} [\mathcal{F}^{\dagger}\mathbb{I}^{\star}\mathcal{F}C_{T_\mathcal{M}}] \\ &-& \Phi_{\gamma} [\mathcal{F}^{\dagger}E_{n}] - \Gamma_{\gamma}[\mathcal{G}^{\dagger}\mathbb{I}^\star\mathcal{F}(\Delta_1, 0, 0, \Delta_4)] \\
&\leq& \lambda_n - \Phi_{\gamma}[\mathcal{P}_{\mathbb{H}_\mathcal{M}^{\perp}}\mathcal{F}^{\dagger}R_{\Sigma^\star}(\mathcal{F}(\Delta))]\\ &-&  \Phi_{\gamma}[\mathcal{P}_{\mathbb{H}_\mathcal{M}^{\perp}}\mathcal{F}^{\dagger}\mathbb{I}^{\star} \mathcal{F}C_{T_\mathcal{M}}]\\ &-&  \Phi_{\gamma}[\mathcal{P}_{\mathbb{H}_\mathcal{M}^{\perp}}\mathcal{F}^{\dagger}E_n]\\ &-&  \Gamma_{\gamma}[\mathcal{P}_{\mathbb{H}_\mathcal{M}[2,3]^{\perp}}\mathcal{G}^{\dagger}\mathbb{I}^\star\mathcal{F}(\Delta_1, 0, 0, \Delta_4)]
\end{eqnarray*}
Here, we used the fact that $\|\mathcal{P}_{T^\perp}(.)\|_2 \leq \|.\|_2$ for a tangent space $T$ of the low-rank matrix variety.

\subsection{\textit{Proof of Proposition 4 (main paper)}}
We must study the rate of convergence of the sample covariance matrix to the population covariance matrix. The following result from \cite{Davidson} plays a key role in obtaining this result.
\begin{proposition}
\label{Proposition}
Given natural numbers $n,p$ with $p \leq n$, Let $\Gamma$ be a $p \times n$ matrix with i.i.d Gaussian entries that have zero-mean and variance $\frac{1}{n}$. Then the largest and smallest singular values $\sigma_1(\Gamma)$ and $\sigma_p(\Gamma)$ of $\Gamma$ are such that:
\begin{eqnarray*}
\max\Bigg\{\text{Prob}[\sigma_{1}(\Gamma) \leq 1 + \sqrt{\frac{p}{n}} + t], \text{Prob}[\sigma_{p}(\Gamma) \leq 1 - \sqrt{\frac{p}{n}} - t]\Bigg\}
\end{eqnarray*} 
\end{proposition}
We now proceed with proving Proposition 4 (main paper). First, note that $\Phi_{\gamma}[\mathcal{F}^{\dagger}E_n] \leq m \|\Sigma_n - \Sigma^\star\|_2$. Using Theorem~\ref{Proposition} and the fact that $\frac{\lambda_n}{m\kappa} \leq 8\psi$ and $n \geq \frac{64\kappa^2(p+q)m^2\psi^2}{\lambda_n^2}$, the following bound holds: $\text{Pr}[ m\|\Sigma_n - \Sigma^\star\|_2 \geq \frac{\lambda_n}{\kappa}] \leq 2\text{exp}\Big\{-\frac{n\lambda_n^2}{128\kappa^2m^2\psi^2}\Big\}$. Thus, $\Phi_{\gamma}[\mathcal{F}^{\dagger}E_n] \leq \frac{\lambda_n}{\kappa}$ with probability greater than  $1 - 2\text{exp}\Big\{-\frac{n\lambda_n^2}{128\kappa^2m^2\psi^2}\Big\}$

\subsection{Consistency of the Convex Program \eqref{eqn:main2} (main paper)}

In this section, we prove the consistency of convex program \eqref{eqn:main2} for estimating a factor model. We first introduce some notation. We define the linear operator:
$\mathcal{\tilde{F}}: \Sp^p \times \Sp^p \rightarrow \Sp^{p}$ and its adjoint ${\tilde{\mathcal{F}}}^{\dagger}: \Sp^{p} \rightarrow \Sp^p \times \Sp^p$ as follows:
\begin{equation}
{\tilde{\mathcal{F}}}(M, K) \triangleq M - K, \qquad \tilde{\mathcal{F}}^{\dagger}(Q) \triangleq (Q,Q)
\label{eqn:OperatorDefs2}
\end{equation}

We consider a population composite factor model \eqref{eqn:composite} $y = \A^\star{x} + \B^\star_u \zeta_u + \epsilon$ underlying a pair of random vectors $(y,x) \in \R^p \times \R^q$, with $\rk(\A^\star) = k_x$, $\B_u^\star \in \R^{p \times k_u}$, and $\colspace({\A^\star}) \cap \colspace({\B^\star_u}) = \{0\}$. As the convex relaxation \eqref{eqn:main2} is solved in the precision matrix parametrization, the conditions for our theorems are more naturally stated in terms of the joint precision matrix $\Theta^\star \in \mathbb{S}^{p+q}, ~ \Theta^\star \succ 0$ of $(y,x)$.  The algebraic aspects of the parameters underlying the factor model translate to algebraic properties of submatrices of $\Theta^\star$.  In particular, the submatrix $\Theta^\star_{yx}$ has rank equal to $k_x$, and the submatrix $\Theta^\star_y$ is decomposable as $D_y^\star - L_y^\star$ with $D^\star_y$ being diagonal and $L^\star_y \succeq 0$ having rank equal to $k_u$.  Finally, the transversality of $\colspace({\A^\star})$ and $\colspace({\B^\star_u})$ translates to the fact that $\colspace(\Theta^\star_{yx}) \cap \colspace(L_y^\star) = \{0\}$ have a transverse intersection. We consider the factor model underlying the random vector $y \in \R^p$ that is induced upon marginalization of $x$. In particular, the precision matrix of $y$ is given by $\tilde{\Theta}^\star_y = D^\star_y - [L^\star_y + \Theta^\star_{yx} (\Theta^\star_{x})^{-1} \Theta^\star_{xy}]$. To learn an accurate factor model, we seek an estimate $(\hat{\tilde{D}}_y, \hat{\tilde{L}}_y)$ from the convex program \eqref{eqn:main2} (main paper) such that $\rk(\hat{\tilde{L}}_y = \rk(L_y^\star + \Theta_{yx}^\star{\Theta_{x}^\star}^{-1}\Theta_{xy}^\star)$, and the errors $\|\hat{\tilde{D}}_y - D^\star_y\|_2,\|\hat{\tilde{L}}_y - [L^\star_y + \Theta^\star_{yx} (\Theta^\star_{x})^{-1} \Theta^\star_{xy}]\|_2$ are small. \\

Following the same reasoning as the Fisher information conditions for consistency of the convex program \eqref{eqn:main}(main paper), A natural set of conditions on the population Fisher information at $\tilde{\Theta}_y^\star$ defined as $\mathbb{I}_y^\star = (\tilde{\Theta}_y^\star)^{-1} \otimes (\tilde{\Theta}_y^\star)^{-1}$ are given by:
\begin{eqnarray}
\mathrm{Assumption~4}&:& \inf_{\mathbb{H}' \in \tilde{U}{(\tilde{\omega}_y)}}\tilde{\chi}({\mathbb{H}}', {\tilde{\Phi}}) \geq \tilde{\alpha}, ~~~ \mathrm{for~some~} \tilde{\alpha} > 0 \label{eqn:FirstFisherCondy} \\[.1in]
\mathrm{Assumption~5}&:& \inf_{\mathbb{H}' \in \tilde{U}{(\tilde{\omega}_y)}}\tilde{\Xi}({\mathbb{H}}') > 0
 \label{eqn:FirstFisherCond3y} \\[.1in]
\mathrm{Assumption~6}&:& \sup_{\mathbb{H}' \in \tilde{U}{(\tilde{\omega}_{y})}}\tilde{\varphi}({\mathbb{H}}') \leq 1-\frac{2}{\tilde{\beta}+1} ~~~ \mathrm{for~some~} \tilde{\beta} \geq 2
\label{eqn:SecondFisherCondy},
\end{eqnarray}
where,
\begin{eqnarray*}
\tilde{\chi}({\mathbb{H}}, \|.\|_{\Upsilon}) &\triangleq& \min_{\substack{Z \in {\mathbb{H}}\\ \|Z\|_\Upsilon= 1}} \|\mathcal{P}_{\mathbb{H}} \tilde{\F}^{\dagger} \mathbb{I}^\star_y\tilde{\F} \mathcal{P}_{\mathbb{H}}(Z)\|_{\Upsilon}\\
\tilde{\Xi}({\mathbb{H}}) &\triangleq& \min_{\substack{Z \in {\mathbb{H}}[2]\\ \|Z\|_2= 1}} \|\mathcal{P}_{\mathbb{H}[2]} \mathbb{I}^\star_y \mathcal{P}_{\mathbb{H}[2]}(Z)\|_{2}\\
\tilde{\varphi}(\mathbb{H}) &\triangleq& \max_{\substack{Z \in \mathbb{H}[2]\\ \|Z\|_{2} = 1}} \|\mathcal{P}_{{{\mathbb{H}}^{\perp}[2]}} \mathbb{I}^{\star}_y\mathcal{P}_{{{\mathbb{H}}[2]}} (\mathcal{P}_{{{\mathbb{H}[2]}}} \mathbb{I}^{\star}_y\mathcal{P}_{{{\mathbb{H}[2]}}})^{-1}(Z)\|_{2} \\
\tilde{U}{(\tilde{\omega}_y)} &\triangleq& \Big\{\mathcal{W} \times T' ~|~ \rho({{T'}}, T({L_{y}^\star + \Theta_{yx}^\star(\Theta_{x}^\star)^{-1}\Theta_{xy}^\star})) \leq \tilde{\omega}_y\Big\}\\
\tilde{\Phi}(D, L) &\triangleq& \max\left\{\|D\|_2,\|L\|_{2}\right\}.
\end{eqnarray*}
Assumption 4 controls the gain of the Fisher information $\mathbb{I}^\star_y$ restricted to appropriate subspaces and Assumption 5 and 6 are in the spirit of irrepresentability conditions. As the variety of low-rank matrices is locally curved around $T(L_y^\star + \Theta_{yx}^\star(\Theta_{x}^\star)^{-1}\Theta_{xy}^\star)$, we control the Fisher information $\mathbb{I}^\star_y$ at nearby tangent spaces $T'$ where $\rho(T', T(L_y^\star + \Theta_{yx}^\star(\Theta_{x}^\star)^{-1}\Theta_{xy}^\star) \leq \tilde{\omega}_y$. We also note that measuring the gains of Fisher information $\mathbb{I}^\star_y$ with the norm $\tilde{\Phi}$ and $\|\cdot\|_{2}$ is natural as these are closely tied with dual norm of the regularizer $\text{trace}(\tilde{L}_y)$ in \eqref{eqn:main2}.

We present a theorem of consistency of the convex relaxation \eqref{eqn:main2} under Assumptions 4, 5 and 6. We let $\sigma$ denote the minimum nonzero singular value of $L_y^\star + \Theta_{yx}^\star(\Theta_{x}^\star)^{-1}\Theta_{xy}^\star$. The proof strategy is similar in spirit to the strategy for proving the consistency of the convex relaxation \eqref{eqn:main} and is left out for brevity. 

\begin{theorem}
\label{theorem:main}
Suppose that there exists $\tilde{\alpha} > 0$, $\tilde{\beta} \geq 2$, $\tilde{\omega}_y \in (0,1)$ so that the population Fisher information $\mathbb{I}^\star_y$ satisfies Assumptions 4, 5 and 6 in \eqref{eqn:FirstFisherCondy},\eqref{eqn:FirstFisherCond3y} and \eqref{eqn:SecondFisherCondy}. Suppose that the following conditions hold:
\begin{enumerate}
\item $n_{} \gtrsim \Big[\frac{\tilde{\beta}^2}{\tilde{\alpha}^2} \Big] (p) $
\item $\tilde{\lambda}_n \sim \frac{\tilde{\beta}}{\tilde{\alpha}}{}\sqrt{\frac{p}{n}}$
\item $\sigma \gtrsim \frac{\tilde{\beta}}{\tilde{\alpha}^5\tilde{\omega}_{y}}{}{}\tilde{\lambda}_n$ 
\end{enumerate}

Then with probability greater than $1-2\exp\Big\{-C\frac{\tilde{\alpha}}{\tilde{\beta}} n\tilde{\lambda}_n^2\Big\}$, the optimal solution $(\hat{\Theta},\hat{\tilde{D}}_y,\hat{\tilde{L}}_y)$ of \eqref{eqn:main2} with i.i.d. observations $\mathcal{D}_{n_{}} = \{y^{(i)}\}_{i = 1}^{n}$ satisfies the following properties:
\begin{enumerate}
\item rank($\hat{\tilde{L}}_y$) = rank(${L}_y^\star + \Theta_{yx}^\star(\Theta_{x}^\star)^{-1}\Theta_{xy}^\star$) \\[.005in]
\item $\|\hat{\tilde{D}}_y - D_y^\star\|_{2} \lesssim \frac{\tilde{\lambda}_n}{\tilde{\alpha}^2}$, $\|\hat{\tilde{L}}_y - L_y^\star - \Theta_{yx}^\star(\Theta_{x}^\star)^{-1}\Theta_{xy}^\star\|_{2} \lesssim \frac{\tilde{\lambda}_n}{\tilde{\alpha}^2}{}$
\end{enumerate}
\end{theorem}